\documentclass[journal]{IEEEtran}
\ifCLASSINFOpdf
\else
\fi
\usepackage{cite}
\usepackage{amsthm}
\usepackage{url}
\usepackage{array}
\usepackage{pgfplots} 
\pgfplotsset{compat=newest} 
\pgfplotsset{plot coordinates/math parser=false} 
\newlength\figureheight 
\newlength\figurewidth  
\usepackage{multirow}
\usepackage{amssymb}
\usepackage{mhchem}
\usepackage{amsmath}
\usepackage{graphicx}
\usepackage{epstopdf}
\usepackage{color}
\usepackage[english]{babel}
\usepackage{verbatim}
\usepackage{tikz}
\newcommand\Tstrut{\rule{0pt}{2.6ex}}         
\newcommand\Bstrut{\rule[-0.9ex]{0pt}{0pt}}
\newcolumntype{C}{>{\centering\arraybackslash}p{3em}}

\newtheorem{theorem}{Theorem}
\newtheorem{lemma}{Lemma}

\newtheorem{remark}{Remark}
\title{
Cooperative look-ahead control for fuel-efficient and safe heavy-duty vehicle platooning
}

\author{Valerio Turri, Bart Besselink, Karl H. Johansson \thanks{The authors are with the ACCESS Linnaeus Centre and Department of Automatic Control, KTH Royal Institute of Technology, Stockholm, Sweden, email: turri@kth.se, bart.besselink@ee.kth.se, kallej@kth.se.} }

\begin{document}

\maketitle
\IEEEpeerreviewmaketitle

\begin{abstract}
The operation of groups of heavy-duty vehicles (HDVs) at a small inter-vehicular distance (known as platoon) allows to lower the overall aerodynamic drag and, therefore, to reduce fuel consumption and greenhouse gas emissions. However, due to the large mass and limited engine power of HDVs, slopes have a significant impact on the feasible and optimal speed profiles that each vehicle can and should follow. Therefore maintaining a short inter-vehicular distance as required by platooning without coordination between vehicles can often result in inefficient or even unfeasible trajectories. In this paper we propose a two-layer control architecture for HDV platooning aimed to safely and fuel-efficiently coordinate the vehicles in the platoon. Here, the layers are responsible for the inclusion of preview information on road topography and the real-time control of the vehicles, respectively. Within this architecture, dynamic programming is used to compute the fuel-optimal speed profile for the entire platoon and a distributed model predictive control framework is developed for the real-time control of the vehicles. The effectiveness of the proposed controller is analyzed by means of simulations of several realistic scenarios that suggest a possible fuel saving of up to $12$\% for the follower vehicles compared to the use of standard platoon controllers.
\end{abstract} 
 
\section{Introduction}
The transportation of goods has been fundamental to the world economic development and the demand for freight transportation, together with the global economy, is expected to increase in the coming years. However, the transport sector, due to the burning of fuel, is responsible for a significant amount of greenhouse gas and \ce{CO2} emissions. In the European Union, the transport sector amounts to roughly 29\% of the total \ce{CO2} emissions and 11\% of these emissions are directly accountable to road freight transportation \cite{eu_pocketbook_2014,CO2_emissions_2012}. Globally, the \ce{CO2} emissions linked to the surface (road and rail) freight transport sector are expected to increase up to 347\% in the next 40 years if no measure is taken \cite{itf_outlook_2015}. In order to contrast this increase and the related impact on the climate change, governments all over the world are agreeing in setting stringent limitations on greenhouse gas emissions connected to road freight transportation \cite{european_commission_2011,epa_2015}. In order to cope with these limitations, heavy-duty vehicle (HDV) manufactures are facing numerous challenges. Furthermore, the expected increase of the oil price \cite{itf_outlook_2015} and the need for maintaining competitiveness require them to design vehicles and technologies that are increasingly fuel-efficient. The fuel cost for an HDV fleet owner, in fact, accounts roughly for the 35\% of the total cost of owning and operating a vehicle \cite{scania_2013}. Therefore even a reduction of a few percent of the fuel consumption would lead to significant saving.

An effective method to reduce fuel consumption and, consequently, greenhouse gas emissions, is HDV platooning. By operating groups of vehicles at small inter-vehicular distances, the overall aerodynamic drag can be reduced. As about one fourth of the HDV fuel consumption is related to aerodynamic drag \cite{hellstrom_2010}, platooning can have a large effect on the fuel consumption. Indeed, experimental results in \cite{phd_alam_2014} and~\cite{bonnet_2010} have shown a reduction in fuel consumption up to 7\%. However, in order to safely operate HDVs at the short inter-vehicular distances required for platooning, automation of the longitudinal dynamics is necessary.

In this work we present a novel control architecture for fuel-efficient and safe HDV platooning. Vehicle platooning is not a new control problem. The first works on vehicle platooning appeared in the sixties, e.g., \cite{peppard_1974,dunbar_2006,naus_2010}. The main focus of these early works was the theoretical study of the dynamics of a string of vehicles with local information, with a particular attention on the study of string stability, i.e., the attenuation of disturbances in position, speed and acceleration along the string of vehicles. The vehicle platooning concept received the first application interest under the Partners for Advanced Transportation Technology (PATH) project \cite{Shladover2007}, where platooning (of passenger vehicles) has been investigated as a means to increase highway throughput. Under this project a control architecture based on vehicle-to-vehicle communication has been developed for the platoon formation and maintenance of fully autonomous vehicles \cite{Horowitz2000}.  Although the environmental aspect was not the focus of the project, noteworthy results on fuel reduction due to HDV platooning have been reported \cite{Browand2004}. Since then, new projects and related publications have appeared with focus on different aspects of HDV platooning such as congestion, safety, fuel-efficiency and user-acceptance \cite{bergenhem2010challenges,tsugawa2013overview,bergenhem2012overview}. 

In the more recent COMPANION project \cite{companion_website}, where this paper finds its place, the fuel-efficiency of HDV platooning is the main focus. The project is  not limited to pursue the efficiency of a single platoon, but rather to create a complete fuel-efficient freight transportation system. This led to the development of a system architecture aimed to divide this complex problem into hierarchical solvable subproblems \cite{phd_alam_2014,phd_liang_2014}. An adaptation of such architecture has three layers, namely the mission planner, the platoon controller and the low-level vehicle controller, defined as follows: the mission planner is responsible for the optimal routing of the HDVs and their synchronization in order to create and dissolve platoons in optimized meeting points. This problem has been addressed in \cite{Larson2015} where the authors propose a distributed framework for the synchronization of single HDVs and platoons on the road network. The platoon controller of each platoon receives from the mission planner the optimal  route and the average speed per link that the platoon should track. Therefore it controls the vehicles' dynamics and it computes the inputs for the low-level vehicle controllers of each HDV. In \cite{Alam2015} a distributed control framework over the platoon suitable for HDV dynamics is presented. However the role of external factors, such as slopes, is not taken into account. 

Because of the large mass and the limited power of HDVs, altitude variations have a significant impact on their behavior. Even small slopes produce such large longitudinal forces on the HDVs that they are often not able to keep constant speed during uphill segments (because of limited engine power) and during downhill segments without applying brakes (because of the significant inertia). Hence it is common that HDVs have to brake and therefore waste energy in order not to overcome the speed limit during downhill sections. This has been addressed in \cite{hellstrom_2009} in order to design a control system that optimizes the fuel consumption of single HDVs driving over hilly roads. In this work the authors showed how, by using look-ahead control based on road topography information and speed limits, it is possible to reduce the fuel consumption of a single HDV up to $3.5$\%. Slopes become more critical in the case of HDVs driving in a platoon formation. In \cite{alam_2013}, the authors point out how the existing look-ahead strategies for single HDVs are not necessarily suitable for a platoon and that a dedicated approach is required. This is due to the fact that the additional requirement of keeping a small inter-vehicular distance between vehicles collides with the fact that HDVs experience significantly different longitudinal forces (e.g., gravity force depending on their mass and current road slope and air drag resistance depending on the distance from the previous vehicle). This appeared evident in the experimental results of a three-vehicle platoon driving on a highway presented in \cite{phd_alam_2014}. Even though in this work the HDVs have similar characteristics, the author highlighted how the use of feedback controllers in particularly hilly sectors of the highway could lead to an increase of the fuel consumption of the follower vehicles compared to the case in which are driving alone. These experimental data are further analyzed in Section~\ref{sec:motivational_experiment} in order to obtain a good understanding of the role of the road gradient on HDV platooning. This analysis provides a motivation for the development of a novel cooperative look-ahead control for HDV platooning with the specific objective of fuel-efficiency, for which some early results have been published in \cite{Turri2014}. Hence, this leads to the following contributions of the current paper.

First, a control architecture for the fuel-efficient and safe control of an HDV platoon is presented. The control architecture is divided into two layers, namely the platoon coordinator and the vehicle controller layers. The platoon coordinator computes the fuel-optimal speed profile for the entire platoon by taking into account information about the road ahead. This optimal profile is communicated to the decentralized vehicle controller layer that safely tracks it and computes the real-time inputs for each vehicle in the platoon. 

Second, two receding horizon strategies within this control architecture are developed. The platoon coordinator relies on a dynamic programming (DP) formulation \cite{book_bellman_1957} that exploits preview information on the road topography and speed limits to compute a speed trajectory defined over space that is safe and fuel-optimal for the whole platoon. Here we emphasize that the platoon coordinator can handle heterogeneous platoons in a systematic way. The vehicle controller layer, instead, is solved though a distributed model predictive control formulation \cite{book_borrelli_2014} that tracks the speed trajectory and a certain gap policy while guarantying fuel-efficiency and safety. More precisely, it is proved that with this architecture no collision will occur within the platoon when up to one vehicle is controlled manually. 

The performance of the proposed control architecture is finally evaluated through extensive simulations motivated by real experimental scenarios and comparisons with existing approaches for speed control and spacing policy are presented.

The rest of the paper is organized as follows. In Section~\ref{sec:motivational_experiment} we analyze the experimental results presented in \cite{phd_alam_2014}. In Section~\ref{sec:model} we present the vehicle and platoon models used in the controllers, whereas in Section~\ref{sec:control_architecture} we introduce the control architecture. The platoon coordinator and the vehicle controller layers are discussed in Sections~\ref{sec:platoon_coordinator} and~\ref{sec:trajectory_tracking}, while their performance is studied in Sections~\ref{sec:analysis_platoon_coordinator},~\ref{sec:analysis_vehicle_controller} and~\ref{sec:analysis_integrated_system}, by means of simulations. Finally, conclusions are stated in Section~\ref{sec:conclusion}.

\section{Motivating experiment} \label{sec:motivational_experiment}
In this section we analyze the experimental results presented in \cite{phd_alam_2014} in order to reach a good understanding of the impact of the road gradient on HDV platooning and motivate the need for the design of a look-ahead control framework for fuel-efficient HDV platooning. 

In this experiment a platoon of three similar HDVs (same powertrain and mass of $37.5$, $38.4$, $39.5$ tons, respectively) is driven over a $45$ km highway stretch between the Swedish cities of Mariefred and Eskilstuna. The topography for this road is displayed in Figure~\ref{fig:matlab_road_topology45}, where the red color highlights the uphill and downhill sections where the slope is too large for a nominal HDV (whose parameters are displayed in Table~\ref{tab:vehicle_parameters} in Section~\ref{sec:analysis_platoon_coordinator}) to maintain a constant speed of $22$ m/s without braking or exceeding the engine power limit. For the considered road, the steep sections represent $23$\% of the total length. Overall, the follower vehicles, by platooning, manage to save $4.1$\% and $6.5$\%, respectively. However in \cite{phd_alam_2014} it is shown that the fuel-efficiency drops significantly in the road sector where the slope is more varying. In this study we therefore analyze the behavior of the first two vehicles while driving over the particularly hilly stretch highlighted in Figure~\ref{fig:matlab_road_topology45} as Sector A for which an increase of the fuel consumption of the second HDV of $4$\% compared to the case of driving alone has been reported. The behavior is shown in Figure~\ref{fig:matlab__experiment}. The first vehicle tracks a reference speed of $21.5$ m/s using cruise control and it switches to braking mode only when the speed limit of $23.6$ m/s is reached, while the second vehicle tracks a headway gap (a distance proportional to its speed) from the first vehicle and it switches to braking mode only when the headway gap reaches a certain threshold (refer to \cite{phd_alam_2014} for a complete characterization of the controllers). In the analyzed sector three critical segments highlighted in Figure~\ref{fig:matlab__experiment} are identified where the use of feedback controls shows its limitations.
\begin{itemize}
\item \textit{Segment 1}: due to the steep downhill the first HDV is not able to maintain the reference speed and therefore it accelerates while coasting. The second vehicle, while trying to track the headway gap policy, follows the same behavior. However, due to the reduced experienced air resistance, during the downhill the second vehicle accelerates more than the first one and, when the critical headway gap is reached, it brakes. In this case the coordination between the accelerations of the two vehicles would have the potential to avoid the undesired braking.
\item \textit{Segment 2}: the headway gap deviates significantly from the reference one, due to a large relative speed at the beginning of the uphill segment and a change of gear during the segment. The second vehicle, in order to reduce the headway gap error, significantly increases the relative speed. Once the critical headway gap is reached, it brakes strongly. In this case, the prediction of the vehicles behaviors would have allowed the second vehicle to reduce the relative speed before reaching the reference headway gap and, therefore, to avoid the undesired braking.
\item \textit{Segment 3}: here the second vehicle shows a more critical behavior compared to the first downhill. In fact, during downhills, the vehicles' actuators work close to saturation (small fueling and small braking) which is not suitable for feedback controller. Therefore in Segment 3 the control state of the second vehicle continues to switch between fueling and braking modes. In this case, the use of a receding horizon framework would have allowed to predict correctly the vehicle behavior depending on the slope and, by taking into account the actuators' saturation, therefore, to obtain a smoother behavior of the vehicle.
\end{itemize}
The analysis of these experimental results provides a strong motivation for the development of a cooperative look-ahead control strategy for HDV platooning based on a receding horizon framework where the road gradient and the vehicles ahead can be explicitly taken into account.  

\begin{figure*}
\begin{center}
\pgfplotsset{every axis/.append style={font=\footnotesize,thin,tick style={ultra thin}}}
\includegraphics{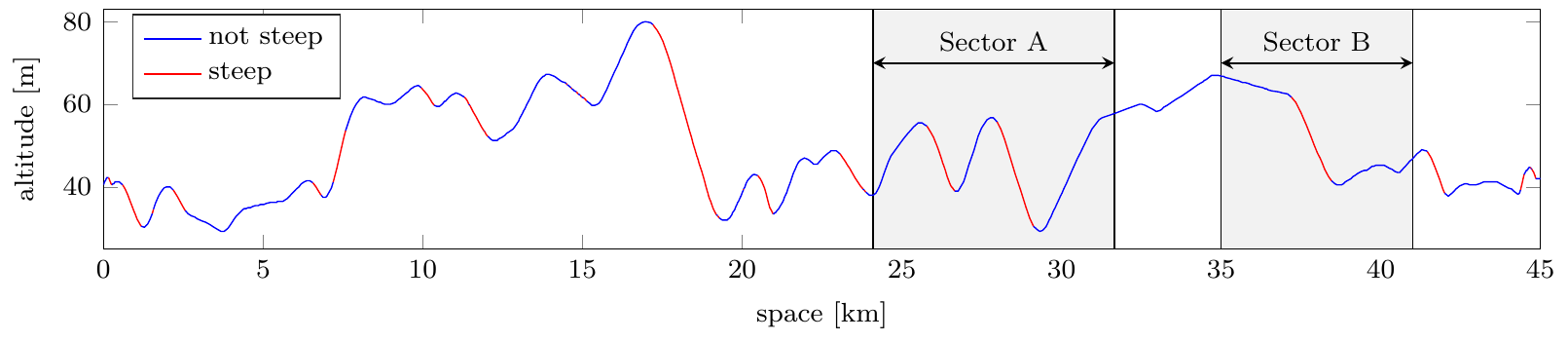}
\caption{Road topography for the $45$ km highway stretch between the Swedish cities of Mariefred and Eskilstuna. The red color highlights the uphill and downhill sections where the slope is too large for the considered HDVs to maintain a constant speed of $22$ m/s without braking or exceeding the engine power limit.}
\label{fig:matlab_road_topology45}
\end{center}
\end{figure*}

\begin{figure}
\begin{center}
\includegraphics{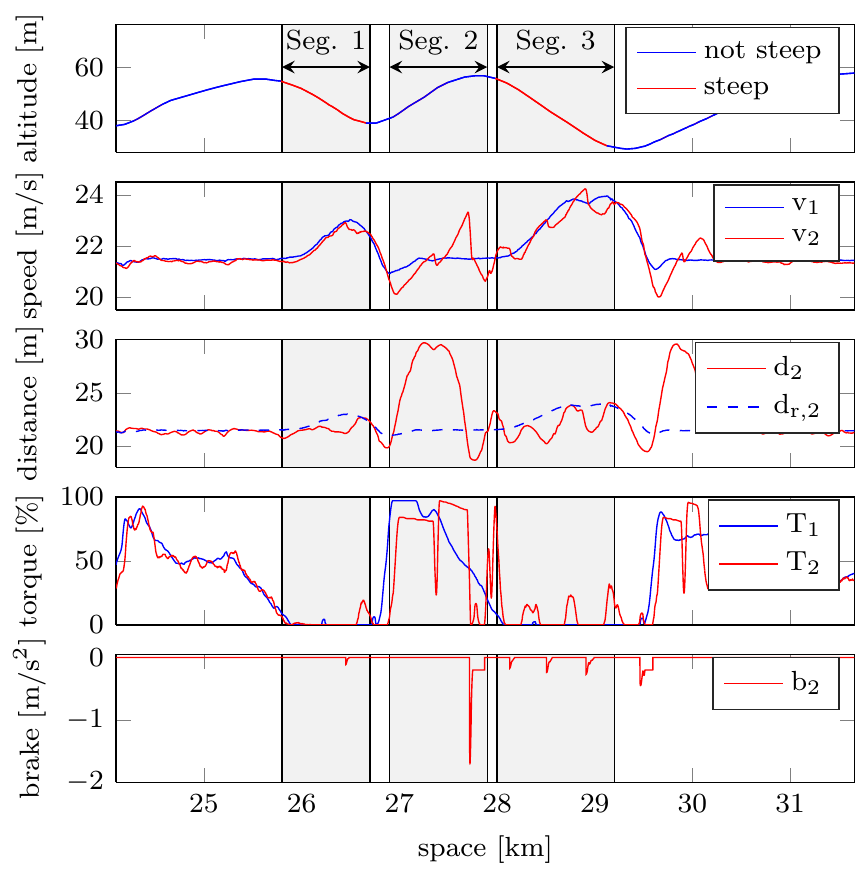}
\caption{Experiment results presented in \cite{phd_alam_2014} relative to the first two vehicles of a three vehicle platoon driving along the Sector A highlighted in Figure~\ref{fig:matlab_road_topology45}. The first plot shows the road topography, whereas the second plot shows the speed of the two vehicles; the third plot shows the real and reference (according to a headway gap policy) between the vehicles; finally the forth and fifth plots shows respectively the normalized engine torque for both vehicles and the normalized braking force for the second vehicle (the braking action of the first vehicle is not available). For additional details, see \cite{phd_alam_2014}.}
\label{fig:matlab__experiment}
\end{center}
\end{figure}

\section{Modeling} \label{sec:model}
HDVs are complex systems with a large number of interacting dynamics. Due to their heavy load, the braking and powertrain systems of an HDV have to generate and transfer extremely high torques. In this section we first present the vehicle system architecture upon which the proposed controller is designed.  Second, we introduce the model of the longitudinal dynamics of a single vehicle and a platoon with a particular focus on the components that play a significant role for the fuel consumption. Finally, we present the fuel model used to estimate the fuel consumption. 

\subsection{Vehicle system architecture} \label{sec:vehicle_system_architecture}
The functioning of an HDV is guaranteed by a large number of system units that communicate with each other through the controller area network (CAN) bus. A simplified control architecture of an HDV is shown in Figure~\ref{fig:system_vehicle_architecture}, where only the system units that are of interest for our work are displayed  \cite{phd_alam_2014}. A more detailed description of a complete vehicle system architecture is given in \cite{johansson2005vehicle}.
\begin{figure}
\begin{center}
\includegraphics{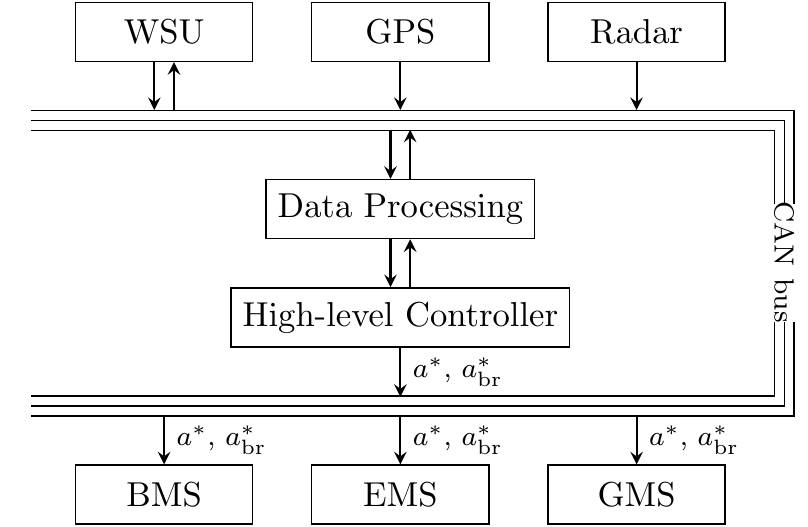}
\caption{Simplified system vehicle architecture.}
\label{fig:system_vehicle_architecture}
\end{center}
\end{figure}
The communication with the outside world relies on the wireless sensor unit (WSU). This unit shares real-time information with the other vehicles within the platoon. The global positioning system (GPS) computes the absolute position of the vehicle, while the radar measures the distance and relative speed between the current vehicle and the preceding one. The real-time state information of the platoon coming from the WSU, the GPS and the radar are fused by the data processing unit and sent to the high-level controller. The high-level controller computes the desired acceleration and a boolean variable defining which low level controller should track it. In particular, the brake management system (BMS) controls the braking actuators, while the engine and gear management systems (EMS and GMS) control the engine, the gearbox and the clutch to provide the requested acceleration. 

\subsection{Vehicle and platoon model}
In this subsection we derive the model of the longitudinal dynamics of a single vehicle and the platoon that is then used in the controller formulation. 
Using Newton's second law, the longitudinal dynamics of a single vehicle can be expressed by: 
\begin{equation}
\begin{aligned}
m_i\dot{v}_i=& \, F_{{\text{e},i}}+F_{{\text{b}},i}+F_{{\text{g}},i}(\alpha(s_i))+F_{{\text{r}},i}+F_{{\text{d}},i}(v_i, d_{i}), \\
\dot{s}_i=& \, v_i,
\end{aligned}
\label{eq:vehicle_model}
\end{equation}
where $v_i$ and $s_i$ are the states of vehicle $i$, respectively, the speed and the longitudinal position, $m_i$ is its mass and $F_{{\text{e}},i}$, $F_{{\text{b}},i}$, $F_{{\text{g}},i}$, $F_{{\text{r}},i}$ and $F_{{\text{d}},i}$ are the forces acting on the vehicle. We collect $v_i$ and $s_i$ in the state vector $x_i=[v_i \; s_i]^\text{T}$). More specifically, $F_{{\text{e}},i}$ and $F_{{\text{b}},i}$ are the control inputs and represent the forces generated by the powertrain and the braking system, respectively. The engine force $F_{{\text{e}},i}$ is characterized in Section~\ref{sec:fuel_model}, while the braking force $F_{{\text{b}},i}$ is assumed to be limited by the road friction and therefore bounded by
\begin{equation}
-m_i \eta_{i} g \mu \leq F_{\text{b},i} \leq 0 ,
\label{eq:model_constraint_braking}
\end{equation}
where $\mu$ and $\eta_{i}$ denote the (positive) road friction coefficient and the braking system efficiency, respectively. Next, $F_{{\text{g}},i}(\alpha(s_i))$ is the force due to the gravity, modeled as 
\begin{equation}
F_{{\text{g}},i}(\alpha(s_i))=-m_i g \sin (\alpha(s_i)),
\end{equation} 
where $g$ is the gravity acceleration and $\alpha(s_i)$ the road slope at position $s_i$. The rolling resistance is represented by $F_{{\text{\text{r}}},i}$ and is modeled as 
\begin{equation}
F_{{\text{r}},i}=-c_r m_i g,
\end{equation}
where $c_r$ is the rolling coefficient. Finally $F_{{\text{d}},i}(v_i, d_{i})$ is the aerodynamic drag, modeled as 
\begin{equation}
F_{{\text{d}},i}(v_i, d_i)=-\tfrac{1}{2}\rho A_v C_D(d_{i}) v_i^2,
\end{equation}
where $\rho$ is the air density, $A_v$ is the cross-sectional area of the vehicle and $C_D$ is the air drag coefficient. In order to take into account the influence of the inter-vehicular distance on the aerodynamic force that plays an essential role in platooning, the drag coefficient $C_D$ is defined as a function of the distance to the previous vehicle $d_i$. This dependence is modeled by 
\begin{equation}
C_D(d_i)=C_{D,0} \left( 1-\frac{C_{D,1}}{C_{D,2}+d_i} \right),
\label{aerodynamic_force_model}
\end{equation}
where the parameters $C_{D,1}$ and $C_{D,2}$ have been obtained by regressing the experimental data presented in \cite{book_hucho_1998}. The effect of the short inter-vehicular distance on the leading vehicle is neglected since it is smaller than one on the follower vehicles. The experimental data and the regression curve are displayed in Figure~\ref{fig:matlab_drag_coefficient}.
\begin{figure}
\begin{center}
\includegraphics{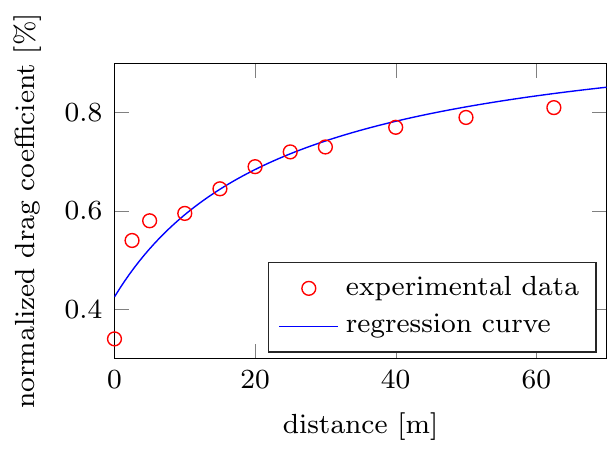}
\caption{Experimental data \cite{book_hucho_1998} and regression curve of the normalized drag coefficient experienced by an HDV as function of the distance to the previous vehicle. }
\label{fig:matlab_drag_coefficient}
\end{center}
\end{figure}

\begin{remark}
In this work we have chosen to model the air drag coefficient on the basis of the experimental data presented in \cite{book_hucho_1998} relative to the second vehicle of a two buses platoon driving at $80$ km/h. In the literature reports on air drag coefficient or fuel consumption measures based on both real experiments \cite{alam_2010,bonnet_2010,Lammert2014} and fluid dynamics simulator \cite{Norrby2014} are presented. They show a reduction of the air drag coefficient for short inter-vehicular distances. How the reduction relates to the inter-vehicular distance varies. This variability has been attributed to the weather condition (e.g, temperature, humidity or wind) and the shape of the vehicles. 
\end{remark}

The model of a platoon of $N_v$ vehicles is given by the combination of equations \eqref{eq:vehicle_model} for $i=1,...,N_v$ and the distance definition:
\begin{equation}
d_i= \left\{ \begin{aligned} 
& \infty ,&\text{if }& i=1,  \\ 
& s_{i-1}-s_i-l_{i-1}, \; &\text{if }& i \geq 2, \end{aligned} \right. ,
\label{distance_definition}
\end{equation}
where $l_i$ denotes the length of vehicle $i$. 

\subsection{Fuel model} \label{sec:fuel_model}
The powertrain is a complex system composed by engine, clutch, gearbox and final gear that allows to transform the fuel's energy into longitudinal force. In this subsection we derive a simple model of the powertrain that captures the intrinsic relation between consumed fuel and generated traction force. In the model derivation we ignore transmission energy losses and the rotational inertia of the powertrain components because they are assumed to be negligible when compared to the vehicle mass.

Engine performance is typically described by the brake specific fuel consumption (BSFC), that defines the ratio between consumed fuel and generated energy for various operation points (i.e., engine speed and generated torque). In Figure~\ref{fig:matlab_engine_map} we show the BSFC map for an HDV engine of $400$ hp \cite{Sandberg_2001_Lic}, where the dotted lines represent the collection of operation points with equal generated power. This map can be easily converted in one that defines the fuel flow $\delta_{i}$ as function of the engine speed $\omega_{i}$ and the generated engine power $P_{i}$, i.e., $\delta_{i}=\phi_{i}(\omega_{i},P_{i})$. By assumption, the engine power $P_{i}$, passing through the clutch, the gearbox and the final gear is completely transferred to the wheels. The rotational speed, instead, changes between the transmission components and is finally transformed into longitudinal speed by the wheels. Ultimately, under the assumption of no longitudinal slip, the vehicle speed $v_i$ can be defined as $v_i=k_i g_i\omega_{i}$, where $k_i$ is a constant gain and $g_i$ is the gear ratio of the gearbox. Therefore the fuel flow can be expressed as a function of the speed $v_i$, the traction force $F_{\text{e},i}$ and the gear ratio $g_i$ as
\begin{equation}
\delta_{i}=\phi_{i}\left(\frac{v_i}{k_i g_i},F_{\text{e},i} v_i \right).
\label{eq:powertrain_model_old}
\end{equation}

\begin{figure}
\begin{center}
\includegraphics{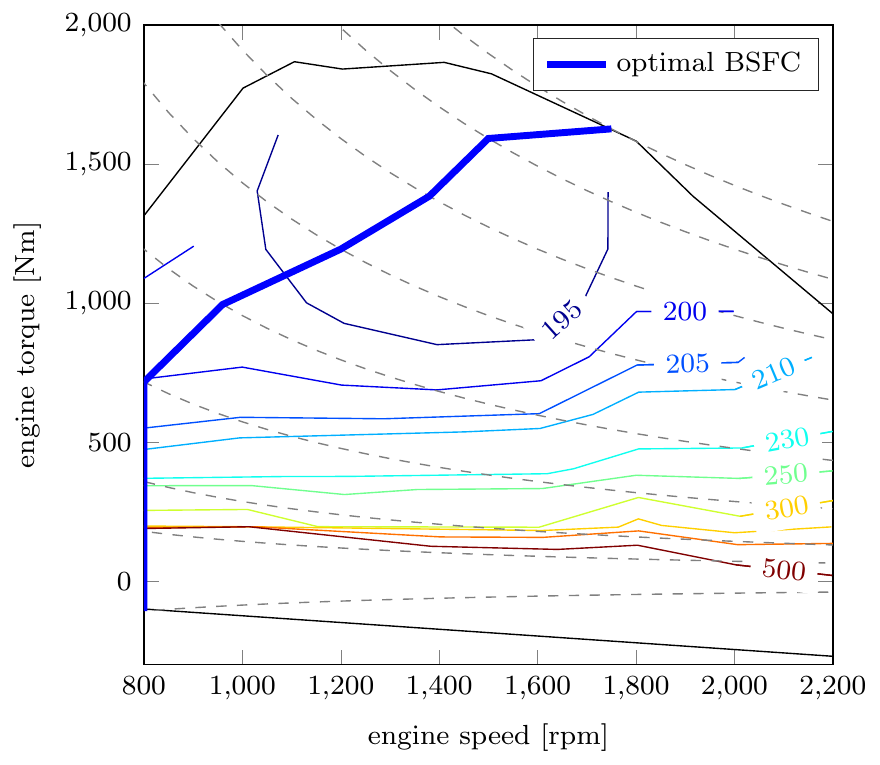}
\caption{BSFC map for a $400$ hp engine regenerated from \cite{Sandberg_2001_Lic}. The plot shows the BSFC expressed in g/kWh as function of the engine speed and torque. The dotted lines represent equal power curves, while the blue thick line represents the collection of the fuel-optimal operation points for various generated powers.}
\label{fig:matlab_engine_map}
\end{center}
\end{figure}

In order to be efficiently used in the control design, the fuel model is further simplified by removing the dependence of the fuel flow $\delta_i$ on the gear ratio $g_i$ through the introduction of an additional assumption: the gear ratio can be changed continuously on a unlimited span and the gear management system chooses the most efficient gear ratio. Hence, we redefine the fuel model as
\begin{equation}
\delta_i= \min_{\omega_i}\phi_{i}\left(\omega_i,F_{\text{e},i} v_i \right)=\phi_{\text{opt},i}(F_{\text{e},i}v_i).
\label{eq:powertrain_model_old2}
\end{equation}
The resulting curve $\phi_{\text{opt},i}(\cdot)$ is depicted in Figure~\ref{fig:fuel_model} and is linearly regressed in order to obtain the fuel model used in the controller design, defined by
\begin{equation}
\delta_i=p_{1,i} F_{\text{e},i} v_i + p_{0,i}.
\label{eq:powertrain_model}
\end{equation}
From this analysis we can also obtain the bounds on the generated power that are independent from the engine speed and the gear ratio:
\begin{equation}
P_{\text{min},i} \leq F_{\text{e},i} v_i \leq P_{\text{max},i}.
\label{eq:model_constraint_power} 
\end{equation} 
In Figure~\ref{fig:fuel_model} the two fuel models in \eqref{eq:powertrain_model_old2} and  \eqref{eq:powertrain_model}, and the correspondent optimal engine speed are displayed. We note that the approximation error is negligible.
\begin{figure}
\begin{center}
\includegraphics{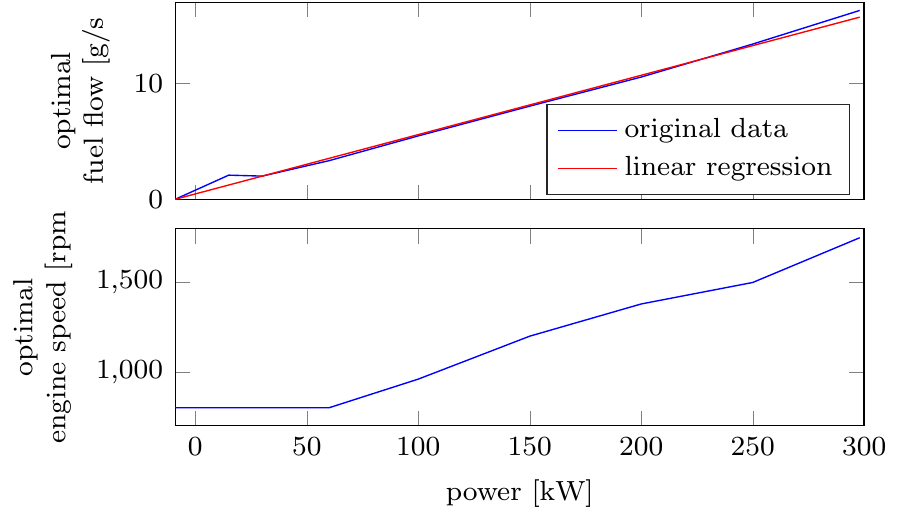}
\caption{The plots show the optimal fuel flow and engine speed as function of the generated power. In the first plot we also display the fuel model expressed in \eqref{eq:powertrain_model} obtained by the regression of the raw data.}
\label{fig:fuel_model}
\end{center}
\end{figure}

\section{Platoon controller architecture}\label{sec:control_architecture}
\begin{figure}
\begin{center}
\vspace{0.5cm}
\begin{tikzpicture}[
planner/.style={rectangle, draw=black, dashed, align=center, minimum height=12mm, minimum width=53.3mm},
db/.style={rectangle, draw=black, dashed, align=center, minimum height=12mm, minimum width=22mm, text width=18mm},
truck/.style={rectangle, draw=black, dashed, align=center, minimum height=12mm, minimum width=22mm},
vehiclelayer/.style={rectangle, draw=black, align=center, minimum height=14mm, minimum width=22mm, text width=18mm},
platoonlayer/.style={rectangle, draw=black, align=center, minimum height=12mm, minimum width=84mm},
signal/.style={semithick,>=stealth},
gap/.style={semithick, >=stealth, dashed}]
\def\dx{1mm};
\node (p) at (0,4) [platoonlayer] {platoon coordinator};
\node (pl) at (-1.55,6) [planner] {mission planner};
\node (db) at (3.1,6) [db] {road database};
\draw[<-,signal] ([xshift=0]{p.north-|db.south}) to node[left] {$\alpha^{\text{s}}$, $v_\text{max}^{\text{s}}$, $v_\text{min}^{\text{s}}$} ([xshift=0]db.south);
\draw[<-,signal] ([xshift=0]{p.north-|pl.south}) to node[left]{$\bar{v}$ } ([xshift=0]pl.south);
\node (t1) at (-3.1,0) [truck] {\includegraphics[width=19mm]{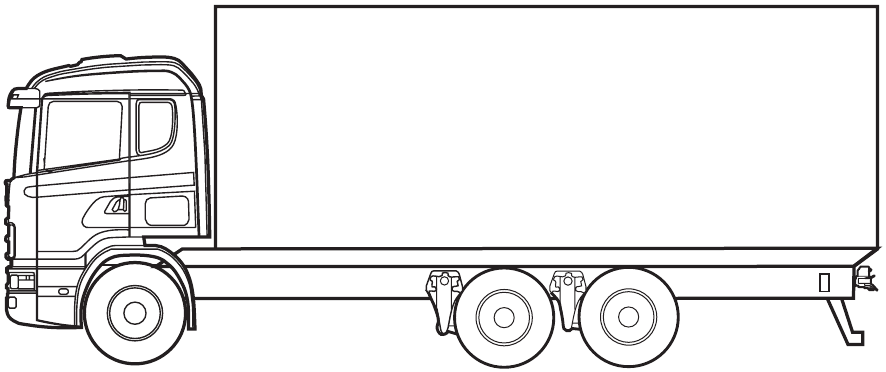}};
\node (v1) at (-3.1,2) [vehiclelayer] {vehicle $1$ controller};
\draw[->,signal] ([xshift=-\dx]v1.south) to node[left]{$a^*_1(t)$} ([xshift=-\dx]t1.north);
\draw[<-,signal] ([xshift=\dx]v1.south) to node[right]{$x_1(t)$} ([xshift=\dx]t1.north);
\draw[->,signal] ([xshift=-\dx]{p.south-|v1.north}) to node[left]{$v^{\text{s},*}(\cdot)\vphantom{x_1}$} ([xshift=-\dx]v1.north);
\draw[<-,signal] ([xshift=\dx]{p.south-|v1.north}) to node[right]{$x_1(t)\vphantom{\bar{v}}$} ([xshift=\dx]v1.north);
\node (t2) at (0,0) [truck] {\includegraphics[width=19mm]{scania_small_standalone.pdf}};
\node (v2) at (0,2) [vehiclelayer] {vehicle $2$ controller};
\draw[->,signal] ([xshift=-\dx]v2.south) to node[left]{$a^*_2(t)$} ([xshift=-\dx]t2.north);
\draw[<-,signal] ([xshift=\dx]v2.south) to node[right]{$x_2(t)$} ([xshift=\dx]t2.north);
\draw[->,signal] ([xshift=-\dx]p.south-|v2.north) to node[left]{$v^{\text{s},*}(\cdot),\tau_{2}\vphantom{x_1}$} ([xshift=-\dx]v2.north);
\draw[<-,signal] ([xshift=\dx]p.south-|v2.north) to node[right]{$x_2(t)\vphantom{\bar{v}}$} ([xshift=\dx]v2.north);
\node (t3) at (3.1,0) [truck] {\includegraphics[width=19mm]{scania_small_standalone.pdf}};
\node (v3) at (3.1,2) [vehiclelayer] {vehicle $3$ controller};
\draw[->,signal] ([xshift=-\dx]v3.south) to node[left]{$a^*_3(t)$} ([xshift=-\dx]t3.north);
\draw[<-,signal] ([xshift=\dx]v3.south) to node[right]{$x_3(t)$} ([xshift=\dx]t3.north);
\draw[->,signal] ([xshift=-\dx]p.south-|v3.north) to node[left]{$v^{\text{s},*}(\cdot),\tau_{3}\vphantom{x_1}$} ([xshift=-\dx]v3.north);
\draw[<-,signal] ([xshift=\dx]p.south-|v3.north) to node[right]{$x_3(t)\vphantom{\bar{v}}$} ([xshift=\dx]v3.north);
\draw[<-,signal] (v2.east-|v3.west) to node[below]{${x_2}(t)$} (v2.east);
\draw[<-,signal] (v1.east-|v2.west) to node[below]{${x_1}(t)$} (v1.east);
\draw[<->,gap] (t1.east) to node[below] {$d_{2}(t)$} (t2.west);
\draw[<->,gap] (t2.east) to node[below] {$d_{3}(t)$} (t3.west);
\end{tikzpicture}
\caption{System architecture for look-ahead HDV platooning.}
\label{fig_architecture}
\end{center}
\end{figure}
The system architecture for the look-ahead HDV platoon controller is shown in Figure~\ref{fig_architecture}. The mission planner suggests routes and platoon opportunities. The platoon coordinator supervises the platoon behavior exploiting road information. The vehicle controllers execute the speed profiles for the individual vehicles. 

The platoon coordinator layer exploits available information on the topography of the planned route to find a fuel-optimal speed profile for the entire platoon, while satisfying the average speed requirement provided by the mission planner. Hereby, in order to capture the dynamics of the road topography, it considers a horizon of several kilometers and takes the constraints of all vehicles in the platoon into account. As a result, it can be guaranteed that every vehicle in the platoon is able to track the required speed profile. A single speed trajectory is computed by the platoon coordinator, representing the speed of the platoon. However, when this speed profile is specified as a function of space (i.e., position on the road) and the inter-vehicle spacing is chosen according to a pure time delay, every individual vehicle in the platoon can track this single speed profile. It is remarked that this layer can typically operate in a receding horizon fashion, providing an updated speed profile roughly every 10 seconds or when the recalculation is needed due to a strong deviation from the original speed profile. Finally, as this layer is not safety-critical and not related to a specific vehicle, it can be implemented in any of the platooning vehicles or even in an off-board road unit. In Section~\ref{sec:platoon_coordinator}, we present a DP approach to formulate and solve the stated problem.

The vehicle controller is responsible for the real-time control of each vehicle in the platoon and is aimed at tracking the desired speed profile as resulting from the platoon coordinator. It also exploits the communication between vehicles of the assumed trajectories to ensure the proper spacing strategy. This layer guarantees the safety of platooning operations to for instance avoid collisions between trucks. Because of the safety critical aspect, this layer is implemented in a distributed fashion in each vehicle of the platoon. More precisely, each vehicle controller runs in the block named high-level controller in the system vehicle architecture shown in Figure~\ref{fig:system_vehicle_architecture}. In Section~\ref{sec:trajectory_tracking} a distributed model predictive control approach for this problem is discussed.

Figure~\ref{fig_formulation_architecture} shows how the optimization problems in the platoon coordinator and the vehicle controllers interact, and their mathematical structure. Note how the platoon coordinator, in order to have a good prediction of the consumed fuel over the horizon, uses an accurate non-linear model of the vehicle, while the vehicle controller layer, in order to enable fast computation necessary for the real-time control of the vehicle, uses a linear vehicle model. 

\begin{figure}
\begin{center}
\vspace{0.5cm}
\begin{tikzpicture}[
platoonlayer/.style={rectangle, draw=black, align=center, minimum height=12mm, minimum width=74mm},
signal/.style={semithick,>=stealth}]
\node (p1) at (0,0) [platoonlayer] {Platoon coordinator  \\  \tiny {\color{white}a}\normalsize \\ $\begin{aligned}
\min \quad & \sum_{i=1}^{N_v} \text{fuel}_i  \\
\text{subj. to} \quad
& \text{Non-linear HDVs model}\\
& \text{Constraints on state and input}\\
& \text{Constraint on avarage speed} \\
& \text{Same speed profile for all HDVs}
  \end{aligned}$};
\node (p2) at (0,-50mm) [platoonlayer] {Vehicle $i$ controller  \\ \tiny {\color{white}a}\normalsize \\ $\begin{aligned}
\min \quad & \text{Deviation from reference profile} \\
\text{subj. to} \quad 
& \text{Linear HDV model}\\
& \text{Constraints on state and input}\\
& \text{Safety constraint} \\
& \text{Soft constraint on braking}
  \end{aligned}$};
\draw[<-,signal] (p2.north-|p2.south) to node[right,text width=3.5cm]{reference speed profile,  gap policy} (p1.south);
\end{tikzpicture}
\caption{Optimal control problems solved in the platoon coordination and vehicle controllers.}
\label{fig_formulation_architecture}
\end{center}
\end{figure}
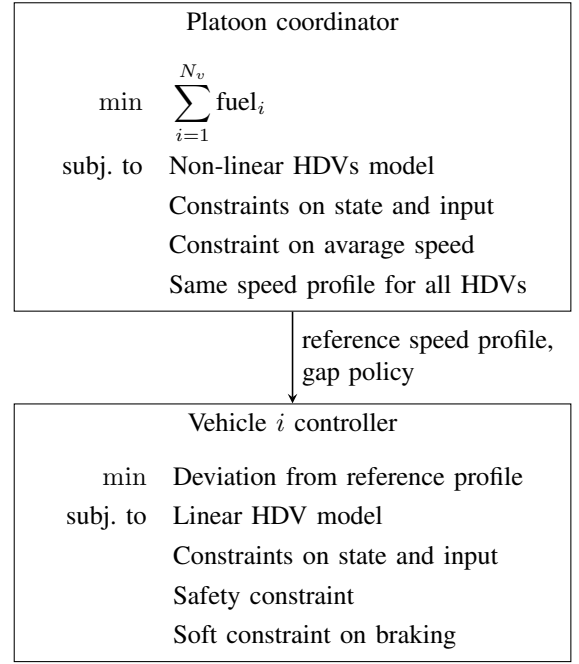

\section{Platoon coordinator} \label{sec:platoon_coordinator}
The platoon coordinator is the higher layer of the platoon control architecture. 
It takes as inputs the average speed requirement $\bar{v}$ from the mission planner and the current vehicles state $x_i(t)$ from their vehicle controllers. By exploiting the available information on the planned route (i.e., slope data $\alpha^{\text{s}}$ and speed limits $v_\text{max}^{\text{s}}$), it generates a unique feasible and fuel-optimal speed profile $v^{\text{s},*}(\cdot)$ defined over space for all the vehicles within the platoon (i.e., $v_i^{\text{s},*}(z)=v^{\text{s},*}(z)$ for $i=j,...,N_v$, where $z$ is the space variable). Furthermore, according to safety criteria, it specifies the time gaps $\tau_i$, defined as the time delay between two consecutive vehicles passing through the same point, i.e., 
\begin{equation}
s_i(t)=s_{i-1}(t-\tau_i).
\label{eq:time_gap}
\end{equation}
Note that this spacing policy is consistent with the requirement that all vehicles have to follow the  same speed profile over space. This can be easily shown by computing the time derivative of the left hand side,
\begin{equation}
\frac{ds_i(t)}{dt}=v_i(t)=v_i^{\text{s}}(s_i(t)),
\label{eq:der_time_gap_lfs}
\end{equation}
and the right hand side of \eqref{eq:time_gap}, 
\begin{equation}
\begin{aligned}
\frac{ds_{i-1}(t-\tau_i)}{dt}&=v_{i-1}(t-\tau_i) \\
&=v_{i-1}^{\text{s}}(s_{i-1}(t-\tau_i))=v_{i-1}^{\text{s}}(s_i(t)),
\label{eq:der_time_gap_rfs}
\end{aligned}
\end{equation}
where $v_i^\text{s}(s)$ denote the speed of vehicle $i$ at space $s$. In fact, by combining the time gap definition \eqref{eq:time_gap} with \eqref{eq:der_time_gap_lfs} and \eqref{eq:der_time_gap_rfs}, we obtain \mbox{$v_i^{\text{s}}(s)=v_{i-1}^{\text{s}}(s)$}.

The coordinator layer is implemented using a DP framework \cite{book_bellman_1957} that runs in closed-loop. The parameters that characterize the DP problem are the discretization space $\Delta s_{\text{DP}}$, the horizon length $H_{\text{DP}}$ and the refresh frequency $f_{\text{DP}}$. We also define the horizon space length as $S_{\text{DP}}=H_{\text{DP}}\Delta s_{\text{DP}}$.

In the coming subsections we introduce all the components of the DP formulation, i.e., the vehicle model, the constraints on the input and states and finally the cost function. 

\subsection{Platoon model} \label{sec:pc_platoon_model}
The platoon coordinator layer uses a discretized version of the vehicle model \eqref{eq:vehicle_model}, where the discretization is carried out in the space domain using the implicit Euler approximation. The discretized vehicle model is: 
\begin{subequations}
\begin{align}
v_{i}^{\text{s}}(z_k)\frac{v_{i}^{\text{s}}(z_k)-v_{i}^{\text{s}}(z_{k-1})}{\Delta s_{\text{DP}}}=
&F^{\text{s}}_{\text{e},i}(z_k)+F^{\text{s}}_{\text{b},i}(z_k) \label{eq:model_ptg}\\
&-m_i g [\sin (\alpha(z_k))+c_r ]  \nonumber \\
&-\tfrac{1}{2}\rho A_v C_{D}(d_i^{\text{s}}(z_k)) (v_i^{\text{s} }(z_k))^2,  \nonumber \\
v_i^{\text{s}}(z_k)\frac{t_i^{\text{s}} (z_k)-t_i^{\text{s}}(z_{k-1})}{\Delta s_{{\text{DP}}}}=&1  \label{eq:model_ptg2},  
\end{align}
\end{subequations}
where $z_k$ is the discretized space variable, $v_{i}^{\text{s}}(z_k)$, $F_{\text{e},i}^{\text{s}}(z_k)$, $F_{\text{b},i}^{\text{s}}(z_k)$ and  $d_i^{\text{s}}(z_k)$ are the speed, the engine and braking forces and the distance to the previous vehicle expressed as function of space, respectively. 

The advantage of using the space discretization is that, by relaxing the average speed requirement, there is no constraint depending on time. The relaxation is done by removing the average speed constraint and introducing instead travel time over the horizon in the cost function, hereby using an appropriate weighting. This allows to ignore the time dynamics and therefore reduce significantly the computational complexity.

A drawback of the space discretization is that the distance definition \eqref{distance_definition} cannot be expressed in the space domain. Instead, the following approximated expression, as function of the current vehicle speed \mbox{$v_i^{\text{s}}(z_k)$}, is used: 
\begin{equation}
d_i^{\text{s}}(z_k)=v_i^{\text{s}}(z_k) \tau_{i}-l_{i-1}.
\end{equation} 

In the DP formulation we refer to \eqref{eq:model_ptg} as \mbox{ $v_{i}^{\text{s}}(z_{k-1})=f_{v,i}^{\text{s}}(v_{i}^{\text{s}}(z_k),{u}_i^{\text{s}}(z_k))$}, where ${u}_i^{\text{s}}(z_k)$ is the input vector defined as ${u}_i^{\text{s}}(z_k)=[F_{\text{e},i}^{\text{s}}(z_k)$, $F_{\text{b},i}^{\text{s}}(z_k) ]^\text{T}$.

\subsection{Model constraints}
The platoon model is constrained by introducing bounds on the input and the speed. 
\subsubsection{Input constraints}
According to \eqref{eq:model_constraint_braking} and \eqref{eq:model_constraint_power}, the engine and braking forces are bounded by the following constraints:
\begin{equation}
\begin{aligned}
P_{\text{min},i}v_i^{\text{s}}(z_k) &\leq F^{\text{s}}_{\text{e},i}(z_k) \leq P_{\text{max},i}v_i^{\text{s}}(z_k),\\
-m_i \eta_{i} g \mu &\leq F^{\text{s}}_{\text{b},i}(z_k) \leq 0.
\end{aligned}
\end{equation}
In the DP formulation, we refer to these constraints as ${u}^{\text{s}}_i(z_k) \in \mathcal{U}^{\text{s}}_i(z_k)$.

\subsubsection{State constraints}
In order to take into account the road speed limits, the speed is bounded by
\begin{equation}
v_{\text{min}}(z_k) \leq v_i^{\text{s}}(z_k) \leq v_{\text{max}}(z_k). 
\label{eq:speed_limit_dp}
\end{equation}
We refer to this constraint as $v_i^{\text{s}}(z_k) \in \mathcal{V}^{\text{s}}(z_k)$. 

Moreover, in order to require all the vehicle to follow the same speed profile, the constraint
\begin{equation}
v_i^{\text{s}}(z_k)=v^{\text{s}}(z_k), \quad i=1,...,N_v.
\label{eq:red_complexity}
\end{equation}
is introduced. The practical effect of this constraint is to reduce the search space of the dynamic programming algorithm to one dimension rather then the number of vehicles in the platoon, enabling fast computation. 

\subsection{Cost function}
\label{sec:cost_function_definition}
The objective of the platoon coordinator layer is to define the optimal speed profile that minimizes the fuel consumption of the whole platoon, while maintaining a certain average speed. This is done by defining the cost function as the weighted sum of two terms: a first term $J_{\text{f}}(v^{\text{s}}(z_{j}),{u}_i^{\text{s}}(z_{j}))$ for $j=k,...,k+H_{\text{DP}}-1$ and $i=1,...,N_v$ representing the fuel amount consumed by the platoon and a second term $J_{\text{t}}(v^{\text{s}}(z_{j}))$ for $j=k,...,k+H_{\text{DP}}-1$ representing the travel time over the horizon, i.e.,
\begin{equation}
J_{\text{DP}}(v^{\text{s}}(z_{j})),{u}^{\text{s}}(z_{j}))=J_{\text{f}}(v^{\text{s}}(z_{j}),{u}_i^{\text{s}}(z_{j}))+\beta J_{\text{t}}(v^{\text{s}}(z_{j})),
\label{eq:cost_function}
\end{equation}
where $\beta$ represents a trade-off weight\footnote{Instead of the constraint on the average speed of Figure \ref{fig_formulation_architecture}, the parameter $\beta$ is tuned to give the desired average time}. The term $J_{\text{f}}(v^{\text{s}}(z_{j}),{u}_i^{\text{s}}(z_{j}))$ is computed by using the fuel model \eqref{eq:powertrain_model}, taking also into account a final term representing the kinematic energy of the platoon at the end of the horizon: 
\begin{equation}
\begin{aligned}
J_{\text{f}}(v^{\text{s}}(z_{j}),{u}_i^{\text{s}}(z_{j})) &\\
= &\sum_{i=1}^{N_v} \sum_{j=k}^{k+H_{\text{DP}}-1}  \Delta s_{\text{DP}} \left(p_{1,i} F^\text{s}_{\text{e},i}(z_{j})  +\frac{p_{0,i}}{v^{\text{s}}(z_{j})} \right) \\
&- \sum_{i=1}^{N_v} p_{1,i} \frac{m_i (v^{\text{s} }(z_{h+H_{\text{DP}}-1}))^2}{2}. \\
\end{aligned}
\nonumber
\end{equation}
The term $J_{\text{t}}(v^{\text{s}}(z_{j}))$ is obtained by using the time model \eqref{eq:model_ptg2}: 
\begin{equation}
J_{\text{t}}(v^{\text{s}}(z_{j}))=\sum_{j=k}^{k+H_{\text{DP}}-1}\frac{\Delta s_{\text{DP}}}{v^{\text{s}}(z_j)}. \nonumber
\end{equation}

\subsection{Dynamic programming formulation} 
\label{sec:dynamic_programming formulation}
We now have all the elements to formulate the DP problem solved in the platoon coordinator:
\begin{subequations} 
\allowdisplaybreaks
\begin{align}
\displaystyle \min_{{u}^{\text{s}}(z_{j})} 
&J_{\text{DP}}(v^{\text{s}}(z_{j}),{u}^{\text{s}}(z_{j}))\\
\text{subj. to} \quad
& v^{\text{s}}_i(z_{j-1})=f_{v,i}^{\text{s}}(v_i^{\text{s}}(z_{j}),{u}_{i}^\text{s}(z_{j})), \\
& {u}^{\text{s}}_i(z_{j}) \in \mathcal{U}^{\text{s}}_i(z_{j}) ,\\
& v_i^{\text{s}}(z_{j})=v^{\text{s}}(z_{j}) \in \mathcal{V}^{\text{s}} (z_{j}),
\label{eq:dp_speed_constraint}  \\
& z_{k}=s_1(t),  \quad \quad  \label{eq:dp_initial1}\\
& v^{\text{s}}(z_{k})=v_1(t), \label{eq:dp_initial2}
\end{align}
\label{eq:dp}
\end{subequations}
$\!\!$for $j=k,...,k+H_{\text{DP}}-1$, where the equations \eqref{eq:dp_initial1} and \eqref{eq:dp_initial2} represent the initial conditions of the DP formulation.

\section{Vehicle controller}
\label{sec:trajectory_tracking}
This section focuses on the distributed model predictive controllers running in the vehicle controller layer. 

Each vehicle controller runs locally. Vehicle $i$ receives the optimal speed profile $v^{\text{s},*}(\cdot)$ and the time gap $\tau_i$ from the platoon coordinator and state information from the preceding vehicle.  By tracking the optimal speed profile and gap policy requirement, and satisfying a safety constraint, it generates the optimal state and input trajectories, respectively  $x_i^*(\cdot|t)$ and $a_i^*(\cdot|t)$, and the desired acceleration $a^*_{i}(t)$ for the vehicle low-level controllers. The parameters that characterize the MPC formulation are the discretization time $\Delta t_{\text{MPC}}$, the horizon steps number $H_{\text{MPC}}$, the refresh frequency $f_{\text{MPC}}$ and the length of the horizon defined as $T_{\text{MPC}}=H_{\text{MPC}} \Delta t_{\text{MPC}}$. 

In the coming subsections we introduce all the components of the MPC formulation, i.e., the vehicle model, the constraints on the input and state, the safety constraint and finally the cost function. 

\subsection{Vehicle model}
In the MPC formulation the vehicle is described by 
\begin{equation}
 x_{i}(t_{j+1}|t_k)= A x_{i}(t_{j}|t_{k})+ B a_i(t_{j}|t_{k}),  
\label{eq:model_vt}
\end{equation}
where 
\begin{equation}
 A \triangleq \left[ 
\begin{matrix}
1 & 0\\
\Delta t_\text{MPC} & 1 
\end{matrix}
\right] , \quad   B\triangleq \left[ \begin{matrix}
\Delta t_\text{MPC} \\
0
\end{matrix} \right]. \nonumber
\label{eq:model_veh1}
\end{equation}
The variables $x_i(t_{j}|t_{k})=[v_i(t_{j}|t_{k}) \; s_i(t_{j}|t_{k})]^\text{T}$ and $a_i(t_{j}|t_{k})$ denote the predicted state (speed and position) and control input (acceleration) trajectories of vehicle $i$ associated to the update time $t_k$, respectively. We also introduce three additional trajectories associated to each update time $t_k$ that will be used later in the MPC formulation:
\begin{itemize}
\item the optimal state trajectory $x_{i}^*(t_{j}|t_k)$,
\item the state reference trajectory $\bar{x}_{i}(t_{j}|t_k)$,
\item the assumed state trajectory $\hat{x}_{i}(t_{j}|t_k)$,
\end{itemize}
for $j=k,...,k+H_{\text{MPC}}-1$ and the corresponding input control trajectories defined likewise. While the predicted and optimal trajectories are function of the optimization variable, the reference and assumed trajectories are pre-computed. More precisely the reference trajectories \mbox{$\bar{x}_{i}(t_{j}|t_k)=[\bar{v}_{i}(t_{j}|t_k) \; \bar{s}_{i}(t_{j}|t_k)]^\text{T}$} and $\bar{a}_{i}(t_{j}|t_k)$ are computed from the reference trajectory $v^{\text{s},*}(\cdot)$ and the current position $s(t_k)$ of the vehicle. In particular, $\bar{s}_{i}(t_{j}|t_k)$ is defined recursively as
\begin{equation}
\bar{s}_{i}(t_{j}|t_k)= \left\lbrace
\begin{aligned}
&s_{i}(t_{j}) ,\; &j =k, \\ &\bar{s}_{i}(t_{j-1}|t_k)+ \Delta t_\text{MPC} \bar{v}^{\text{s},*}(\bar{s}_{i}(t_{j-1}|t_k)),\; &j >k,
\end{aligned}
\right.
\nonumber
\end{equation}
while $\bar{v}_{i}(t_{j}|t_k)$ is defined as
\begin{equation}
\bar{v}_{i}(t_{j}|t_k)=\bar{v}^{\text{s},*}(s_i(t_{j}|t_k)).
\nonumber
\end{equation}
The control input reference trajectory $\bar{a}_{i}(t_{j}|t_k)$ is defined as finite differences of $\bar{v}_{i}(t_{j}|t_k)$, i.e.,
\begin{equation}
\bar{a}_{i}(t_{j}|t_k)= (\bar{v}_{i}(t_{j+1}|t_k)-\bar{v}_{i}(t_{j}|t_k)) \Delta t_\text{MPC}. 
\nonumber
\end{equation}
The assumed state and control input trajectories are computed from the optimal and real trajectories of the vehicle as 
\begin{equation}
\hat{x}_{i}(t_{j}|t_k)= \left\lbrace
\begin{aligned}
&x_{i}(t_{j}) ,\; &&j <k, \\ &{x}_{i}^*(t_{j}|t_{k-1}),\; &&k \leq j <k+H_\text{MPC},
\end{aligned}
\right.
\label{eq:safety_definition_assumed}
\end{equation}
and $\hat{a}_{i}(t_{j}|t_k)$ likewise. As mentioned at the beginning of this section, each vehicle communicates the assumed trajectory $\hat{x}_{i}(t_{j}|t_k)$ to the follower vehicle. In this case, the use of the optimal trajectory computed the previous step reflects the assumption of a maximum communication delay of $\Delta t_\text{MPC}$.

\subsection{Input and model constraints}
In order to take into account the bounds on the braking force \eqref{eq:model_constraint_braking} and the engine power \eqref{eq:model_constraint_power}, as done in the platoon coordinator layer, the control input $a_i$ is bounded by the following non-linear constraint: 
\begin{equation}
-\eta_i \mu g+ \frac{F_\text{ext}(x_i,\hat{s}_{i-1})}{m_i} \leq a_i \leq \frac{P_{i,\text{max}}}{m_i v_i}+\frac{F_\text{ext}(x_i,\hat{s}_{i-1})}{m_i},
\label{eq:constraint_on_input}
\end{equation}
where $F_\text{ext}(x_i,\hat{s}_{i-1})$ denotes the summation of the external forces acting on the vehicle and is defined as
\begin{equation}
\begin{aligned}
F_\text{ext}(x_i,\hat{s}_{i-1})=&-m_i g (\sin(\alpha(s_i))+c_r ) \\&-\tfrac{1}{2}\rho A_v C_{D}(\hat{s}_{i-1}-s_{i}-l_i) v_i^2.
\label{eq:F_ext_definition}
\end{aligned}
\end{equation}
The control input is additionally bounded by a soft constraint in order to allow braking only if necessary, i.e., when the safety constraint (see section~\ref{sec:safety_constraint}) is activated or the braking is required by the platoon coordinator. This is formulated as follows:
\begin{equation}
a_i+\epsilon_i \geq \min(a_{\text{c},i},\bar{a}_i), \; \epsilon_i \geq 0, 
\label{eq:constraint_on_input_soft}
\end{equation}
where $\epsilon_i$ is the softening variable and $a_{\text{c},i}$ is the coasting acceleration (i.e., no braking and fuel injection) and is defined as: 
\begin{equation}
a_{\text{c},i}=\frac{P_{i,\text{min}}}{m_i v_i}+\frac{F_\text{ext}(x_i,\hat{s}_{i-1})}{m_i}.
\label{eq:a_costing_definition}
\end{equation}

In the MPC formulation we refer to the constraint \eqref{eq:constraint_on_input},\eqref{eq:F_ext_definition} as ${a}_i(t_j|t_k) \in \mathcal{A}_i(x_i(t_j|t_k))$ and to the soft constraint \eqref{eq:constraint_on_input_soft},\eqref{eq:a_costing_definition} as $a_i(t_j|t_k)+\epsilon_i(t_j|t_k) \in \mathcal{A}_{\text{e},i}(x_i(t_j|t_k))$. 

The speed is bounded according to the constraint \eqref{eq:speed_limit_dp} as
\begin{equation}
v_{\text{min}}({s}_i(t_j|t_k)) \leq v_i(t_j|t_k) \leq v_{\text{max}}({s}_i(t_j|t_k)). \nonumber
\end{equation}
In the MPC formulation, we refer to this constraint as $v_i(t_j|t_k) \in \mathcal{V}(s_i(t_j|t_k))$.

\subsection{Safety constraint}
\label{sec:safety_constraint}
The platoon is intended to operate on standard highways where other vehicles are present. The designed controller therefore should be able to cope with cases where the platoon behavior deviates from the predicted one because of internal disturbances (e.g., gear shift) or external events (e.g., related to the traffic situation or a vehicle cutting into the platoon). In this section we focus on the safety problem, leaving to further work the study of how such  events should be handled (i.e., autonomously or switching to manual driving). 

The platoon is considered safe if, whatever a vehicle in the platoon does, there exists an input for all the follower vehicles such that collision can be avoided. The safety of the platoon is guaranteed by ensuring that the state of each vehicle lies within a safety set and it is firstly studied by considering two adjacent vehicles and later extended to the entire platoon. In here we consider the following vehicle continuous dynamics:
\begin{equation}
\dot{\tilde{x}}_i=
\left[\begin{array}{c}
\dot{\tilde{v}}_i \\
\dot{\tilde{s}}_i 
\end{array}\right]
=f(\tilde{x}_i,\tilde{a}_i)=
\left[\begin{array}{c}
\tilde{a}_i \\
\tilde{v}_i
\end{array}\right],
\label{eq:safety_model}
\end{equation}
where $\tilde{v}_i$, $\tilde{s}_i$ and $\tilde{a}_i$ are the speed, position and acceleration of vehicle $i$, respectively. 

Let us now focus on the dynamics of two adjacent vehicles described by
\begin{equation}
\begin{aligned}
\left[\begin{array}{c}
\dot{\tilde{x}}_{i-1} \\
\dot{\tilde{x}}_i
\end{array}\right]&=F(\tilde{x}_{i-1},\tilde{x}_i,\tilde{a}_{i-1},\tilde{a}_i) =
\left[\begin{array}{c}
f(\tilde{x}_{i-1},\tilde{a}_{i-1}) \\
f(\tilde{x}_{i},\tilde{a}_{i})
\end{array}\right],
\label{eq:safety_model2}
\end{aligned}
\end{equation}
where the acceleration of the current vehicle $\tilde{a}_i$ is the control input, while the acceleration of the previous vehicle $\tilde{a}_{i-1}$ is the exogenous input that can be regarded as a disturbance. We also introduce the admissible set \mbox{$\tilde{\mathcal{X}}=\{[\tilde{x}_{i-1}^\text{T} \, \tilde{x}_{i}^\text{T}]^\text{T}: \tilde{v}_{i-1} \geq 0, \, \tilde{v}_i \geq 0, \, \tilde{s}_{i-1}-\tilde{s}_{i} \geq l_{i-1}\}$} as the set of all admissible states, where $l_i$ denotes the length of vehicle $i$. In order to obtain a closed form of the safety set, the following conservative approximations of the the exogenous and control inputs are introduced:
\begin{subequations}
\begin{align}
\tilde{a}_{i-1} \in \mathcal{A}^\text{p}(\tilde{x}_{i-1}) &= \begin{cases} [\underline{a}_{\text{min},i},\overline{a}_{\text{max},i}], &\; \text{if } \tilde{v}_{i-1}>0, \\ 
[0,\overline{a}_{\text{max},i}], & \; \text{if } \tilde{v}_{i-1}=0, \end{cases} \label{eq:safety_constraint_disturbance}\\ 
\tilde{a}_i \in \mathcal{A}^\text{f}(\tilde{x}_i)&=\begin{cases} [\overline{a}_{\text{min},i},\underline{a}_{\text{max},i}], &\; \text{if } \tilde{v}_i>0, \\ 
[0,\underline{a}_{\text{max},i}], &\; \text{if } \tilde{v}_i=0, \end{cases} \label{eq:safety_constraint_control_input}
\end{align}
\end{subequations}
where $\underline{a}_{\text{min},i}$, $\overline{a}_{\text{min},i}$, $\underline{a}_{\text{max},i}$ and $\overline{a}_{\text{max},i}$ are lower and upper bounds on the minimum and maximum possible accelerations of vehicle $i$, respectively. Such bounds are computed under reasonable assumptions on the vehicles and road properties, i.e., the vehicles' speed is limited ($0 \leq \tilde{v}_i \leq v_\text{max}$), the admissible vehicles' weight is bounded (\mbox{$m_i \in M$}) and the road slope $\alpha$ is bounded ($|\alpha|  \leq \alpha_\text{max}$). For example, the bounds $\underline{a}_{\text{min},i}$ and $\overline{a}_{\text{min},i}$ can be computed as follows:
\begin{subequations}
\begin{align}
\underline{a}_{\text{min},i}=& \min_{0 \leq v \leq v_\text{max},m \in M, |\alpha|  \leq \alpha_\text{max},d \geq0} a_{\text{min},i}(v,m,\alpha,d), \nonumber \\
\overline{a}_{\text{min},i}=& \max_{0 \leq v \leq v_\text{max},m \in M,|\alpha|  \leq \alpha_\text{max},d \geq0} a_{\text{min},i}(v,m,\alpha,d), \nonumber 
\end{align}
\end{subequations}
where 
\begin{equation}
a_{\text{min},i}(v,m,\alpha,d)=-\mu \eta_i g -g \sin(\alpha) -c_r - \tfrac{1}{2} \rho A_v C_D (d) v^2 \nonumber
\end{equation}
and $-\mu \eta_i g$ represents the maximum braking capacity of vehicle $i$. Note that, due to the definition of the bounds and because of the dominance of the $-\mu \eta_i g$ term in the definition of $a_{\text{min},i}$, the following inequalities hold:
\begin{subequations}
\begin{align}
\underline{a}_{\text{min},i} &\leq \overline{a}_{\text{min},i} \leq 0, \label{eq:safety_relation_amin} \\
\underline{a}_{\text{max},i} &\leq \overline{a}_{\text{max},i}. \label{eq:safety_relation_amax} 
\end{align}
\label{eq:safety_relation} 
\end{subequations}

In order to guarantee the safety of the subsystem \eqref{eq:safety_model2}, we should guarantee that the state $[\tilde{x}_{i-1}^\text{T} \, \tilde{x}_i^\text{T}]^\text{T}$ always lies in a safety set $\mathcal{S}$ included in $\tilde{\mathcal{X}}$, for any admissible trajectory of the previous vehicle. 
\begin{figure}
\begin{center}
\includegraphics{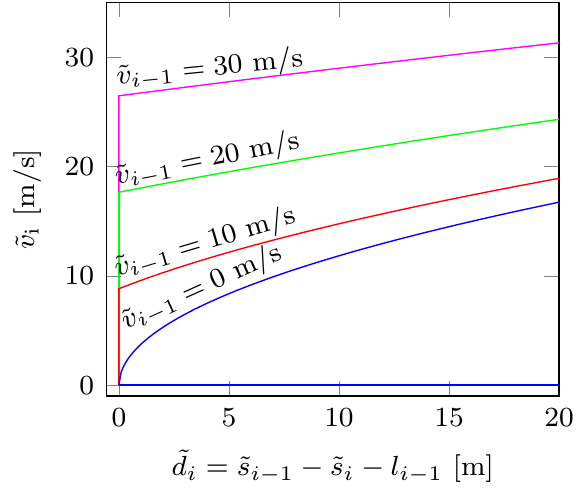}
\caption{Contour plot of the safety set boundary $\partial \mathcal{S}$. The variable $\tilde{d}_i$ denotes the distance between the two adjacent vehicles.}
\label{fig:matlab_safety_set}
\end{center}
\end{figure}
We now define the safety set $\mathcal{S} \subseteq \tilde{\mathcal{X}}$, displayed in \mbox{Figure~\ref{fig:matlab_safety_set}}, as 
\begin{equation}
\mathcal{S}=\{[\tilde{x}_{i-1}^\text{T} \, \tilde{x}_{i}^\text{T}]^\text{T}: g_j(\tilde{x}_{i-1},\tilde{x}_i) \geq 0, \;j=1,...,4\},
\label{eq:safety_definition_set}
\end{equation}
where
\begin{equation}
\begin{aligned}
&g_1(\tilde{x}_{i-1},\tilde{x}_i)=\tilde{s}_{i-1}-\tilde{s}_{i}-l_{i-1}-\frac{\tilde{v}_{i-1}^2}{2\underline{a}_{\text{min},i-1}}+\frac{\tilde{v}_i^2}{2\overline{a}_{\text{min},i}}, \\ 
&g_2(\tilde{x}_{i-1},\tilde{x}_i)=\tilde{s}_{i-1}-\tilde{s}_{i}-l_{i-1} ,  \\
&g_3(\tilde{x}_{i-1},\tilde{x}_i)=\tilde{v}_{i-1} ,  \\
&g_4(\tilde{x}_{i-1},\tilde{x}_i)=\tilde{v}_i 
\end{aligned}
\label{eq:safety_definition_g}
\end{equation}
and we state the following result:
\begin{lemma}
Given the dynamic system \eqref{eq:safety_model2} and the constraints \eqref{eq:safety_constraint_disturbance} and \eqref{eq:safety_constraint_control_input} on the exogenous and control inputs respectively, there exists a control law $\tilde{a}_i=\phi([\tilde{x}_{i-1}^\text{T} \, \tilde{x}_i^\text{T}]^\text{T}) \in \mathcal{A}^\text{f}(\tilde{x}_i)$ such that for all $[\tilde{x}_{i-1}^\text{T}(t_0) \, \tilde{x}_i^\text{T}(t_0)]^\text{T} \in \mathcal{S}$ and $\tilde{a}_{i-1} \in \mathcal{A}^\text{p}(\tilde{x}_{i-1})$, the condition $[\tilde{x}_{i-1}^\text{T}(t) \, \tilde{x}_i^\text{T}(t)]^\text{T} \in \mathcal{S}$ holds for all $t \geq t_0$. In other words, $\mathcal{S}$ is a robust controlled invariant set  \cite{blanchini_1999}.
\end{lemma}
\begin{proof}
By using Nagumo's theorem for robust controlled invariant sets \cite{blanchini_1999}, the lemma can be proved by showing that for all $[\tilde{x}_{i-1}^\text{T} \, \tilde{x}_i^\text{T}]^\text{T} \in \partial \mathcal{S}$ (defined as the boundary of $\mathcal{S}$) there exists an $\tilde{a}_i \in \mathcal{A}^\text{f}$ such that, for all $\tilde{a}_{i-1} \in \mathcal{A}^\text{p}$, the relation
\begin{equation}
\nabla g_j(\tilde{x}_{i-1},\tilde{x}_i)^\text{T}F(\tilde{x}_i,\tilde{x}_{i-1},\tilde{a}_{i-1},\tilde{a}_i) \geq 0
\end{equation}
holds for all $j$ such that $g_j(\tilde{x}_{i-1},\tilde{x}_i)=0$. Because of the structure of the problem, the control input $\tilde{a}_i$ is chosen as
\begin{equation}
\tilde{a}_i=\begin{cases} \overline{a}_{\text{min},i}, \; &\text{if } \tilde{v}_i>0, \\ 0, \; &\text{if } \tilde{v}_i=0,\end{cases}
\label{eq:safety_tilda_a_2_def}
\end{equation} 
for any $[\tilde{x}_{i-1}^\text{T} \, \tilde{x}_i^\text{T}]^\text{T} \in \partial \mathcal{S}$ and $\tilde{a}_{i-1} \in \mathcal{A}^\text{p}(\tilde{x}_{i-1})$. We organize the proof by considering the $[\tilde{x}_{i-1}^\text{T} \, \tilde{x}_i^\text{T}]^\text{T} \in \partial \tilde{\mathcal{S}}$ defined by the activation of each $g_j(\tilde{x}_{i-1},\tilde{x}_i) \geq 0$:
\begin{itemize}
\item for $[\tilde{x}_{i-1}^\text{T} \, \tilde{x}_i^\text{T}]^\text{T}$ such that $g_1(\tilde{x}_{i-1},\tilde{x}_i)=0$, and $g_j(\tilde{x}_{i-1},\tilde{x}_i) \geq 0$, for $j=2,3,4$,
\begin{subequations}
\begin{align}
\nabla g_1(\tilde{x}_{i-1},\tilde{x}_i)^\text{T}&F(\tilde{x}_{i-1},\tilde{x}_i,\tilde{a}_{i-1},\tilde{a}_i) \nonumber \\
=&\left(1- \frac{\tilde{a}_{i-1}}{\underline{a}_{\text{min},i-1}} \right)\tilde{v}_{i-1} -\left(1- \frac{\tilde{a}_i}{\overline{a}_{\text{min},i}} \right)\tilde{v}_i, \nonumber \\
=&\left(1- \frac{\tilde{a}_{i-1}}{\underline{a}_{\text{min},i-1}} \right)\tilde{v}_{i-1} \geq 0, \nonumber
\end{align}
\end{subequations}
where the equality and inequality hold because of how $\tilde{a}_i$ is defined by \eqref{eq:safety_tilda_a_2_def} and $g_3(\tilde{x}_{i-1},\tilde{x}_i) \geq 0$.
\item  for $[\tilde{x}_{i-1}^\text{T} \, \tilde{x}_i^\text{T}]^\text{T}$ such that $g_2(\tilde{x}_{i-1},\tilde{x}_i)=0$, and $g_j(\tilde{x}_{i-1},\tilde{x}_i) \geq 0$, for $j=1,3,4$,
\begin{subequations}
\begin{align}
\nabla g_2(\tilde{x}_{i-1},\tilde{x}_i)^\text{T}F(\tilde{x}_{i-1},\tilde{x}_i,\tilde{a}_{i-1},\tilde{a}_i)&=\tilde{v}_{i-1}-\tilde{v}_i  \geq 0, \nonumber
\end{align}
\end{subequations}
where the inequality holds by noticing that the combination of $g_1(\tilde{x}_{i-1},\tilde{x}_i) \geq 0$, $g_2(\tilde{x}_{i-1},\tilde{x}_i) = 0$ and the relation \eqref{eq:safety_relation_amin} gives $\tilde{v}_{i-1} \geq (\underline{a}_{\text{min},i}/\overline{a}_{\text{min},i}) \tilde{v}_i$.
\item for $[\tilde{x}_{i-1}^\text{T} \, \tilde{x}_i^\text{T}]^\text{T}$ such that $g_3(\tilde{x}_{i-1},\tilde{x}_i)=0$, and $g_j(\tilde{x}_{i-1},\tilde{x}_i) \geq 0$, for $j=1,2,4$,
\begin{subequations}
\begin{align}
\nabla g_3(\tilde{x}_{i-1},\tilde{x}_i)^\text{T}F(\tilde{x}_{i-1},\tilde{x}_i,\tilde{a}_{i-1},\tilde{a}_i)&=\tilde{a}_{i-1} \geq 0, \nonumber
\end{align}
\end{subequations}
where the inequality holds because of \eqref{eq:safety_constraint_disturbance}. The same can be verified in a similar way for $[\tilde{x}_{i-1}^\text{T} \, \tilde{x}_i^\text{T}]^\text{T}$ such that $g_4(\tilde{x}_{i-1},\tilde{x}_i)=0$ and \mbox{$g_j(\tilde{x}_{i-1},\tilde{x}_i) \geq 0 \text{ for } j=1,2,3$.} \qedhere
\end{itemize}
 \end{proof} 

The choice of the safety set guarantees that the follower vehicle can react to the emergency braking maneuver of its predecessor, such that both vehicles come to a standstill without colliding. We now extend the result in Lemma 1 to the safety of the whole platoon. More precisely, we proof that whatever a vehicle does, there exists an input for all the follower vehicles, such that collision can be avoided. This is formalized by the following theorem:
\begin{theorem}
Consider a vehicle with index $i_0<N_\text{v}$ and all its follower vehicles $i \in \mathcal{I}=\{i_0+1,...,N_\text{v}\}$ satisfying the dynamics in \eqref{eq:safety_model}. Then, there exists a control law $\tilde{a}_i = \phi(\tilde{x}_i,\tilde{x}_{i-1}) \in \mathcal{A}^\text{f}(\tilde{x}_i), \, i \in \mathcal{I}$ such that for all $[\tilde{x}_{i-1}^\text{\rm{T}}(t_0) \, \tilde{x}_i^\text{\rm{T}}(t_0)]^\text{\rm{T}} \in \mathcal{S}$ and $\tilde{a}_{i_0} \in \mathcal{A}^\text{p}(\tilde{x}_{i_0})$, the condition $[\tilde{x}_{i-1}^\text{\rm{T}}(t) \, \tilde{x}_i^\text{\rm{T}}(t)]^\text{\rm{T}} \in \mathcal{S}$ holds for all $t \geq t_0$ and all $i \in \mathcal{I}$.
\end{theorem}

\begin{proof}
The application of Lemma 1 for $i=i_0+1$ proves the existence of an input $\tilde{a}_i \in \mathcal{A}^\text{f}(\tilde{x}_i)$ that ensures that $[\tilde{x}_{i-1}^\text{T}(t) \, \tilde{x}_i^\text{T}(t)]^\text{T} \in \mathcal{S}$ for all $t \geq t_0$. Then, by noting that $\mathcal{A}^\text{f}(\tilde{x}_i) \subseteq \mathcal{A}^\text{p}(\tilde{x}_i)$ according to \eqref{eq:safety_relation}, it follows that $\tilde{a}_i \in \mathcal{A}^\text{p}(\tilde{x}_i)$.The theorem is then proven by induction over the vehicle index, hereby repetitively applying Lemma 1. 
\end{proof}

This result is adapted to the MPC formulation in order to guarantee the safety of the platoon. More precisely, each vehicle, knowing the assumed state trajectory of the vehicle ahead, can compute the safety set for its own predicted state. By taking into account that, according to the definition of the assumed state in \eqref{eq:safety_definition_assumed}, the real state of the previous vehicle is know with a one step delay, the safety set $\mathcal{S}$ translates to the following safety constraint on each follower vehicle state: 
\begin{equation}
\begin{aligned}
s_i(t_{j+1}|t_k)-&\frac{v_i^2(t_{j+1}|t_k)}{2 \overline{a}_{\text{min},i}} \\
\leq &\hat{s}_{i-1}(t_{j-1}|t_k)-\frac{\hat{v}_{i-1}^2(t_{j-1}|t_k)}{2 \underline{a}_{\text{min},i}}-l_{i-1},
\end{aligned}
\label{eq:safety_constraint}
\end{equation}
for $i=2,...,N_v$. Note that for safety purpose only the safety constraint for $j=k$ is necessary. In fact it guarantees that, if at the update time $t_k$ the current state of each follower vehicle is safe, then it is going to be safe also at the update time $t_{k+1}$. However, the safety constraint for $j>k$ gives optimal trajectories that are safe over the whole horizon and therefore produces a smoother and more fuel-efficient behavior of the platoon. In the MPC formulation, we refer to the safety constraint \eqref{eq:safety_constraint} as $f_\text{safe}(x_i(t_{j+1}|t_k)) \geq 0$.


\subsection{Cost function}
The objective of the vehicle controller layer is to follow the optimal trajectory and the gap policy requirement provided by the platoon coordinator layer. This can be formulated by introducing the following cost function: 
\begin{equation}
\begin{aligned}
J_i^{\text{MPC}}({x}_i(\cdot,t_{k}),&{a}_i(\cdot,t_{k}),{\epsilon}_i(\cdot,t_{k}))\\
=\sum_{j=k}^{k+H_{\text{MPC}}-1}    &\left|\left|x_i(t_j|t_k)-\hat{x}_{i-1}(t_{j-T_i}|t_k)\right|\right|_{\zeta_i Q} \\
+& \left|\left|x_i(t_j|t_k)-\bar{x}_i(t_j|t_k)\right|\right|_{(1-\zeta_i) Q} \\
+&\left|\left|a_i(t_j|t_k)-\bar{a}_i(t_j|t_k)\right|\right|_{R} \\
+&\left|\left|\epsilon_i(t_j|t_k)\right|\right|_{P}, \nonumber
\end{aligned}
\label{eq:cost_function}
\end{equation}
where 
\begin{equation}
\zeta_i= \begin{cases} 0, \; &\text{if } i=1, \\ \bar{\zeta}, \; &\text{if } i=2,...,N_v \end{cases}
\end{equation}
and $T_i$ represents the discretized version of the time gap $\tau_i$ (i.e., $T_i=\lfloor \tau_i/ \Delta t_\text{MPC} \rfloor$). The parameters $Q$, $R$ and $\bar{\zeta} \in [0, \, 1]$ can be chosen in order to have a good trade-off between reference trajectory, gap policy tracking and actuators excitation. The weight $P$ related to the softening-variable of the constraint \eqref{eq:constraint_on_input_soft} is chosen relatively large such that only the activation of the safety constraint $f_\text{safe}(x_i(t_{j+1}|t_k)) \geq 0$ can require a significant braking force.

\subsection{Model predictive control formulation} \label{sec:model_predictive_control_forumation}
We now have all the elements to formulate the MPC problem:
\begin{subequations} 
\allowdisplaybreaks
\begin{align}
\displaystyle \min_{{a}_i(\cdot,t_{k}),{\epsilon}_i(\cdot,t_{k})} \;
&J_i^{\text{MPC}}({x}_i(\cdot,t_{k}),{a}_i(\cdot,t_{k}),{\epsilon}_i(\cdot,t_{k}))\\
\text{subj. to} \quad
&  x_{i}(t_{j+1}|t_k)= A x_{i}(t_{j}|t_{k})+ B a_i(t_{j}|t_{k}), \\
& {a}_i(t_j|t_k) \in \mathcal{A}_i(x_i(t_j|t_k)), \label{eq:MPC_formulation_a} \\
& a_i(t_j|t_k)+\epsilon_i(t_j|t_k) \in \mathcal{A}_{\text{e},i}(x_i(t_j|t_k)),\label{eq:MPC_formulation_a2} \\
& v_i(t_j|t_k) \in \mathcal{V} (s_i(t_j|t_k)), \label{eq:MPC_formulation_v}\\
& f_\text{safe}(x_i(t_{j+1}|t_k)) \geq 0, \text{ if } i\geq 2, \label{eq:safe_const}  \\
& \epsilon_i(t_j|t_k) \geq 0, \\
& {x}_i(t_{k}|t_k)=x_{i}(t), \label{eq:mpc_init} 
\end{align}
\end{subequations}
where $j=k,...,k+H_{\text{MPC}}-1$ and \eqref{eq:mpc_init} represents the initial condition of the MPC problem. For implementation purpose the state-dependent constraint set in \eqref{eq:MPC_formulation_a}, \eqref{eq:MPC_formulation_a2} and \eqref{eq:MPC_formulation_v} will be replaced respectively by \mbox{$\mathcal{A}_i(\hat{x}_i(t_j|t_k))$}, \mbox{$\mathcal{A}_{\text{e},i}(\hat{x}_i(t_j|t_k))$} and \mbox{$\mathcal{V} (\hat{s}_i(t_j|t_k))$}. Taking into account that the safety constraint \eqref{eq:safe_const} is quadratic and convex, the MPC problem can be recasted into a quadratic constraint quadratic programming (QCQP) problem for which efficient solvers exist. 

The output of the vehicle controller is the desired acceleration $a_{i}^*(t_k)$ (defined as $a_{i}^*(t_k)=a_{i}^*(t_k|t_k)$, where $a_{i}^*(\cdot|t_k)$ is the optimal input trajectory resulting from the MPC) and a boolean variable $a_{\text{br},i}$ defined as 
\begin{equation}
a_{\text{br},i}=\begin{cases} 1, \; \text{if } a_{i}^*(t_k) < a_{\text{c},i}^*(t_k|t_k), \\
0, \; \text{if } a_{i}^*(t_k) \geq a_{\text{c},i}^*(t_k|t_k),  \end{cases}
\end{equation}
that states if the desired acceleration should be tracked by the BMS or the EMS. 

\section{Performance analysis of the platoon coordinator} \label{sec:analysis_platoon_coordinator}
In this section we analyze the performance of the platoon coordinator (as presented in Section \ref{sec:platoon_coordinator} and shown in Figures \ref{fig_architecture} and \ref{fig_formulation_architecture}) by focusing on fuel-efficiency. We compare its performance with other standard controller setups. To make the analysis independent from the low-level tracking strategy, we assume in this section that the HDVs can follow exactly the speed trajectories and spacing policies defined by the high-level controllers.

\subsection{Experiment setup} \label{sec:fuel_efficiency_analysis}
\label{sec:analysis_platoon_coordinator_setup}
The comparison is done by using as benchmark the scenario introduced in Section~\ref{sec:motivational_experiment}. We therefore consider a platoon of two HDVs driving over the $45$ km road stretch shown in Figure~\ref{fig:matlab_road_topology45} and investigate the controller performance for both homogeneous and heterogeneous platoons. The performance metrics chosen to compare the different control configurations are the energy and the fuel consumed by the HDVs. In some comparisons the consumed energy is preferred over the consumed fuel because it can be directly related to the energies dissipated by the various forces (i.e., gravity, rolling, drag and braking forces). 

The control configurations considered in the comparisons include three control strategies and three gap policies. In detail, the following control strategies are considered:
\begin{itemize}
\item cruise control (CC): the first vehicle keeps the constant reference speed $v_\text{CC}$ on low-grade slopes. If the uphill slope is too large to maintain constant speed, the engine generate the maximum power $P_{\text{max}}$ until the speed reaches $v_\text{CC}$ again. If the downhill slope is too large to maintain constant speed without braking, the engine coasts (i.e., does not inject any fuel, generating therefore the minimum power $P_\text{min}$) until the speed reaches $v_\text{CC}$ again. However, if the HDV reaches the speed limit $v_\text{max}$, the brakes are activated in order not to overcome it;
\item look-ahead control (LAC): the first vehicle exploits the slope information of the road ahead in order to minimize its own fuel consumption. 
\item cooperative look-ahead control (CLAC): the first vehicle follows the speed profile generated by the platoon coordinator proposed in this paper. 
\end{itemize}
The following gap policies are considered:
\begin{itemize}
\item space gap (SG): the second vehicle keeps a constant distance $d_\text{SG}$ from the first vehicle;
\item headway gap (HG): the second vehicle keeps a constant headway time $\tau_\text{HG}$ from the first vehicle, i.e., it keeps a distance proportional to its speed ($d_\text{HG}(t)=\tau_\text{HG} v_i(t)$);
\item time gap (TG):  the second vehicle keeps a constant time gap $\tau_\text{TG}$ from the first vehicle according to \eqref{eq:time_gap}.
\end{itemize}
In order to be able to maintain exactly the desired gap policies as previously assumed, the second vehicle is allowed to overcome the theoretical maximum engine power $P_{\text{max},i}$ and to brake if necessary. In addition, in order to obtain a fair comparison it is ensured, by tuning the trade-off parameter $\beta$ of the LAC and CLAC formulations (see \eqref{eq:cost_function}), that the different control strategies have the same average speed $\bar{v}$ and the parameters $d_\text{SG}$, $\tau_\text{HG}$ and $\tau_\text{TG}$ are chosen such that the vehicles in the different gap policies have the same distance when driving at constant speed $\bar{v}$ (i.e., $d_\text{SG}=\bar{v}\tau_\text{HG}=\bar{v} \tau_{\text{TG}}-l_1$). Finally in order to remove the influence of the residual kinematic energy, the initial and final speeds are constrained to be the same in all the controller configurations.

\subsection{Fuel-efficiency analysis for different control strategies} \label{sec:performance11}
In this section we present the results of the platoon behavior for the three different control strategies, while keeping a TG policy ($\tau_\text{TG}=1.4$ s). In the first part, as in the motivational example of Section~\ref{sec:motivational_experiment},  we focus on the homogeneous platoon scenario, while in the second part we consider two heterogeneous platoons (i.e., where the second vehicle is respectively heavier and lighter than the leading one).

\begin{table}[h]
\caption{Vehicle's paramters}
\label{tab:vehicle_parameters}
\begin{center}
\begin{tabular}{l|c}
Parameter & Value \\
\hline 
mass ($m_i$) & $40$ t \\
length ($l_i$)& $18$ m \\
roll coefficient ($c_r$)& $3\times 10^{-3}$ \\
vehicle cross-sectional area ($A_v$)& $10 \, \text{m}^2$ \\
maximum engine power ($P_{\text{max},i}$) & $298$ kW \\
minimum engine power ($P_{\text{min},i}$)& $-9$ kW 
\end{tabular}
\end{center}
\end{table}

We now consider a platoon of two identical vehicles, whose parameters values are displayed in Table~\ref{tab:vehicle_parameters}. We start the comparison by analyzing the comprehensive bar diagram displayed in Figure~\ref{fig:matlab_comparison_control_bars} representing the energy consumed by each vehicle of the platoon for the three control strategies (the corresponding fuel consumption is displayed in the central column of Table~\ref{tab:percentage_fuel_saved_1}). This energy is normalized respect to the energy consumed by a single vehicle driving alone using CC. The consumed energy is additionally split into various components representing the energy dissipated by each force, namely the gravity, roll, drag and braking force. We can first notice how the second vehicle, for all the control strategies, consumes less energy compared to the first one, due to the significant reduction of the energy associated to the drag force. Second, comparing the three control strategies, we can observe how the use of the LAC allows both vehicles to save energy, respectively $3.5$\% and $6.4$\% compared to the use of the CC. Instead, by switching from the LAC to the CLAC, the first vehicle consumes $0.1$\% more energy, while the second one saves $3.7$\% of energy; therefore the platoon, given by the union of the two vehicles, saves $3.6$\% of energy. This result is in line with our expectation since the LAC optimizes the fuel consumption of the first vehicle, while the CLAC targets the reduction of the fuel consumption of the entire platoon. Consequently, the saving of the CLAC strategy with respect to the LAC strategy are expected to increase for platoons of more vehicles. Going into the details of the various consumed energy components, first we notice that the gravity and roll energy components are the same for both vehicles for all the considered control strategies. This is due to the fact that the gravity energy depends only on the difference of altitude between the initial and final points, while the roll energy only depends on the driven distance that is the same by experiment design specification. The drag energy, instead, is significantly different for the two vehicles because of its dependence on the distance to the preceding vehicle, while it is approximately the same for the different control strategies. What significantly changes between the different control strategies is the energy dissipated by braking. 

\begin{figure}
\begin{center}
\includegraphics{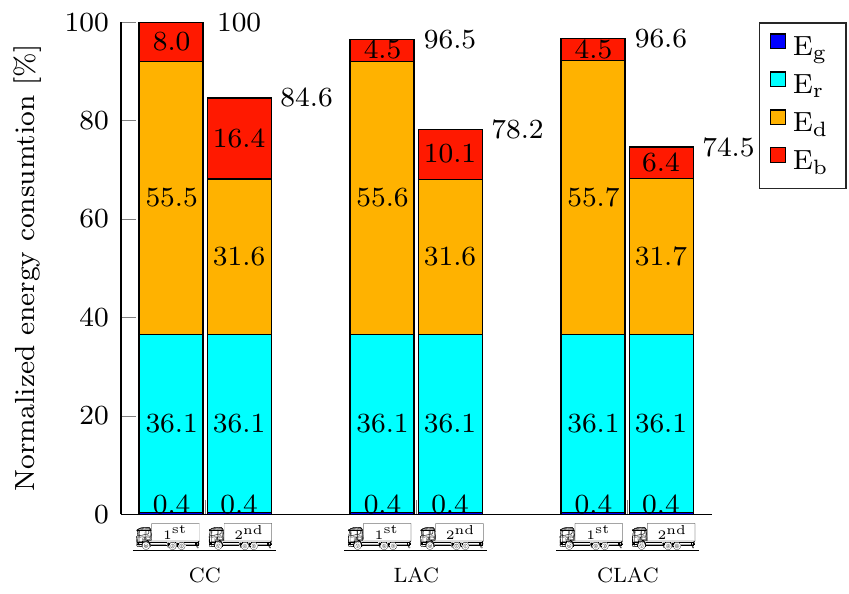}
\caption{Comparison of the energy consumed by each vehicle of a platoon ($m_1=m_2=40$ t), for the three control strategies, namely CC, LAC and CLAC, while keeping a TG policy,  driving along the $45$ km road displayed in Figure~\ref{fig:matlab_road_topology45}. Each bar represents the consumed energy normalized respect to the the energy consumed by a single vehicle driving alone using CC. The consumed energy is split into various components representing the energy dissipated by each force, namely the gravity ($E_\text{g}$), roll ($E_\text{r}$), drag ($E_\text{d}$) and braking ($E_\text{b}$) force.}
\label{fig:matlab_comparison_control_bars}
\end{center}
\end{figure}

In order to understand the role of the control strategies in the braking usage in Figure~\ref{fig:matlab_comparison_control_plot} we show part of the simulation results corresponding to the road highlighted as segment B in Figure~\ref{fig:matlab_road_topology45}. In this study we have chosen to focus on a downhill section because this is where the braking action is taking place. The comparison of the platoon behaviors follows:
\begin{itemize}
\item CC: during the downhill, starting from speed $v_\text{CC}$, the first HDV accelerates while coasting due to the large road grade. In the meantime the second vehicle has to brake slightly in order to maintain the time gap and compensate the reduced drag force compared to the first vehicle. At $38.1$ km, in order not to overcome the speed limit, both vehicles  need to brake significantly;
\item LAC: by exploiting the topography information of the road ahead, the first vehicle reduces its speed before the downhill by anticipating the coasting phase such that the speed limit is reached only when the slope grade is small enough to stop accelerating while coasting and therefore it avoids braking. The second vehicle, as in the CC case, has to brake slightly while the first vehicle is coasting but it also avoids the significant braking phase at the end of the downhill;
\item CLAC: since in this case the optimization is done considering the fuel consumption of both vehicles, with respect to the LAC case the first vehicle starts to loose speed earlier before the downhill. This allows it to fuel sightly during the downhill, allowing the second vehicle to coast meanwhile and, as in the LAC case, to reach the speed limit only when the slope grade is small enough to stop accelerating while coasting. In this case both HDVs do not need to brake. 
\end{itemize}
Note that, in the case of longer downhill segments, the lower speed bound does not allow the vehicle to decrease the speed enough before the downhill in order not to hit the upper speed limit during the downhill. This is why in some sections of the $45$ km benchmark road, in the LAC case, the first vehicle and, in the CLAC case, both vehicles still need to brake. 

\begin{figure}
\begin{center}
\includegraphics{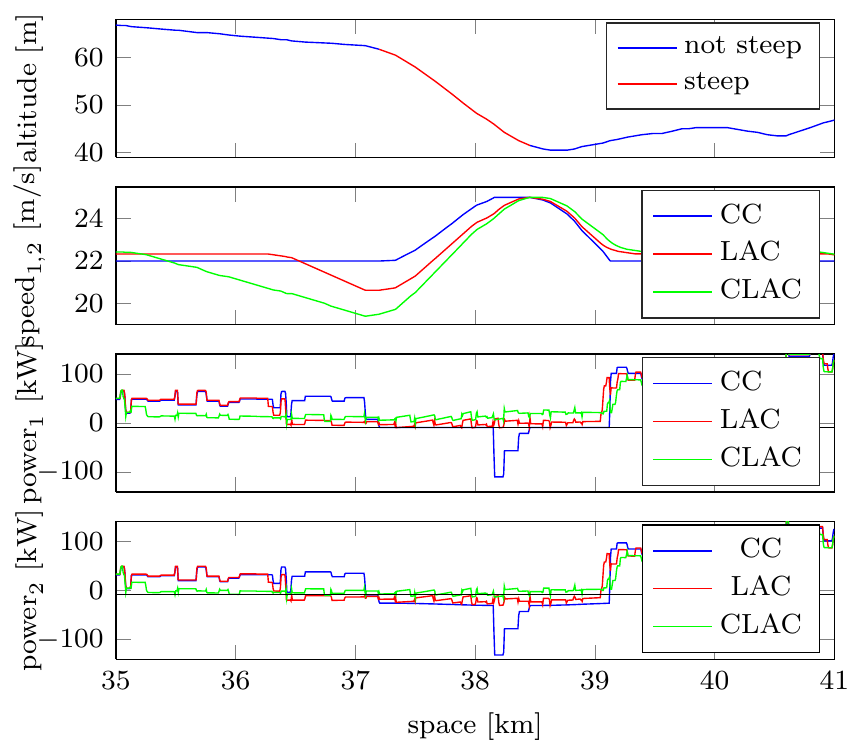}
\caption{Comparison of the behavior of an homogeneous platoon (i.e., $m_1=m_2=40$ t) for three different control strategies, namely CC, LAC and CLAC, while keeping a TG policy,  driving along the Sector B displayed in Figure~\ref{fig:matlab_road_topology45}. The first plot shows the road altitude, where the red color is used to highlight the sections too steep to keep a constant speed of $22$ m/s while respecting the power limit and avoiding braking; the second plot shows the speed profiles for the three control strategies followed by both vehicles (because of the SG policy); finally the third and forth plots show the summation between the generated power by the engine and the braking systems for the two vehicles and three control strategies; the black lines in such plots define the theoretical minimum and maximum engine power, respectively $P_{\text{min},i}$ and $P_{\text{max},i}$ (hence if the power crosses the lower power limit $P_{\text{min},i}$, the respective vehicle is braking).}
\label{fig:matlab_comparison_control_plot}
\end{center}
\end{figure}

So far we have considered the case of an homogeneous platoon. What we want to investigate now is the the role of the different control strategies in the case of heterogeneous platoons. To answer this question, in Table~\ref{tab:percentage_fuel_saved_1}  we have reported  the normalized fuel consumption for the cases of two heterogeneous platoons and the same homogeneous platoon previously considered. More in detail, the HDVs have the same powertrain, but their masses vary between $35$, $40$ and $45$ t. Analyzing the table we can notice how in the case of a heavier second vehicle the CLAC allows to save $10.8$\% of fuel compared to the CC, while, in the case of an lighter second vehicle, it allows to save $5.4$\%. However if we only analyze the last row we can note how, with the use of the CLAC, the order of the vehicles does not significantly change the normalized fuel consumption. 

Concluding, the proposed controller (CLAC) has a significant impact on the reduction of the energy and fuel consumption. In detail, the majority of the fuel saving is related to the reduction of energy dissipated by braking during the downhill sections. The impact of such a controller grows in the case of heavier follower vehicle. 

\begin{table}[h]
\caption{\rm Normalized fuel consumption of the vehicles in the platoon for different control strategies and scenarios (vehicle weights). The fuel is normalized respect to the fuel consumed by the respective HDV driving alone using CC. For the acronyms explanation refers to Section~\ref{sec:fuel_efficiency_analysis} [\%].}
\label{tab:percentage_fuel_saved_1}
\begin{center}
\begin{tabular}{C|C C|C C| C C}
& \multicolumn{2}{c|}{} & \multicolumn{2}{c|}{} & \multicolumn{2}{c}{}  \\
& \multicolumn{2}{c|}{\includegraphics[scale=1]{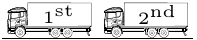}} & 
\multicolumn{2}{c|}{\includegraphics[scale=1]{platoon_two_vehicle_table_num}} & 
\multicolumn{2}{c}{\includegraphics[scale=1]{platoon_two_vehicle_table_num}}  \\
mass & $35$ t & $45$ t & $40$ t & $40$ t & $45$ t & $35$ t \Bstrut\\
   \hline 
CC & $100.0$ & $90.2$ & $100.0$ & $86.3$ &$100.0$ & $82.1$ \Tstrut\\
LAC & $97.6$ & $84.9$ & $96.9$ & $80.6$ & $96.3$ & $77.2$ \\
CLAC & $97.8$ & $78.0$ & $97.0$ & $77.4$ & $96.4$ & $76.7$ \\
\end{tabular}
\end{center}
\end{table}

\subsection{Fuel-efficiency analysis for different gap policies}
In the previous analysis we have always considered a TG policy. The aim of this section is to compare the platoon performance for different gap policies, namely the SG, HG and TG policies, while keeping the same control strategy (in the analysis we have considered CC). Note that in order to be able to follow the required gap policy the second vehicle is allowed to exceed the maximum engine power. In this section we only focus on the homogeneous platoon, since the results for an heterogeneous platoon are qualitatively the same. In Figure~\ref{fig:matlab_comparison_gaps_bars} we show the comprehensive bar diagram representing the normalized energy consumed by each vehicle of the platoon for the three gap policies, while using CC as control strategy. Since the first vehicle uses the same control strategy, the energy consumption defers only for the second vehicle. It is interesting to notice that, similarly to the comparisons done in the previous section, the main difference between the energy consumption of the second vehicles is related to the energy dissipated by braking. More in detail the HG policy allows the second vehicles to save $1$\% over the SG policy, while the TG policy allows to save an additional $1.6$\% of energy. In order to understand the role of the gap policy on the braking energy, we show the platoon behavior driving over a synthetic hill composed by an uphill section with constant slope grade, a flat section and a downhill section with constant slope grade. The platoon behavior for such a hill is shown in Figure~\ref{fig:matlab_comparison_gaps_plot}. Analyzing the second vehicle behavior for each gap policy, the following can be observed:
\begin{itemize}
\item TG policy: as argued in Section~\ref{sec:platoon_coordinator}, the time gap allows the vehicles to follow the same speed profile over space. That means that the generated forces and therefore the generated powers (because of the equal speed result) are equivalent except for a reduction of the air drag component in  the second vehicle. Therefore the power generated by the second vehicle, as can be observed in Figure~\ref{fig:matlab_comparison_gaps_plot}, is approximately a biased equivalent of that one generated by the first vehicle. 
\item SG policy: the space gap, instead, requires the vehicles to follow the same speed profile over time. An interesting consequence can be observed, for example, at the beginning of the uphill section shown in Figure~\ref{fig:matlab_comparison_gaps_plot}; as soon as the first vehicle enters the uphill section and decelerates because of limited engine power, the second vehicle, which is still in the flat section, has to brake in order to respect the space gap requirement. In general, excluding the offset given by the drag power, every time the slope increases (in Figure~\ref{fig:matlab_comparison_gaps_plot}, entering the uphill and leaving the downhill sections), the second vehicle has to generate less power than the first vehicle, while every time the slope decreases (in Figure~\ref{fig:matlab_comparison_gaps_plot}, leaving the uphill and entering the downhill sections) the second vehicle has to generate more power than the first vehicle. As a consequence, the second vehicle has respectively to brake and to exceed the power limit in order to follow the required SG policy. 
\item HG policy: the headway gap can be considered as a trade-off between a time gap and a space gap. In fact, for example, as soon as the first vehicle enters the uphill section and starts to decelerate, the distance between the two vehicles is allowed to decrease, but this decrease is not as fast as in the case of the time gap. 
\end{itemize}
The results obtained by the analysis of the platoon behavior in the case of the synthetic hill are valid also in the case of the original scenario. In conclusion, the time gap allows to save more energy compared to the space and headway gaps. In addition the time gap allows all the vehicles to follow the same speed trajectory in space and therefore it scales well with the number of vehicles in the platoon. The complete results for the normalized fuel consumption are reported in Table~\ref{tab:percentage_fuel_saved_2}.

\begin{figure}
\begin{center}
\includegraphics{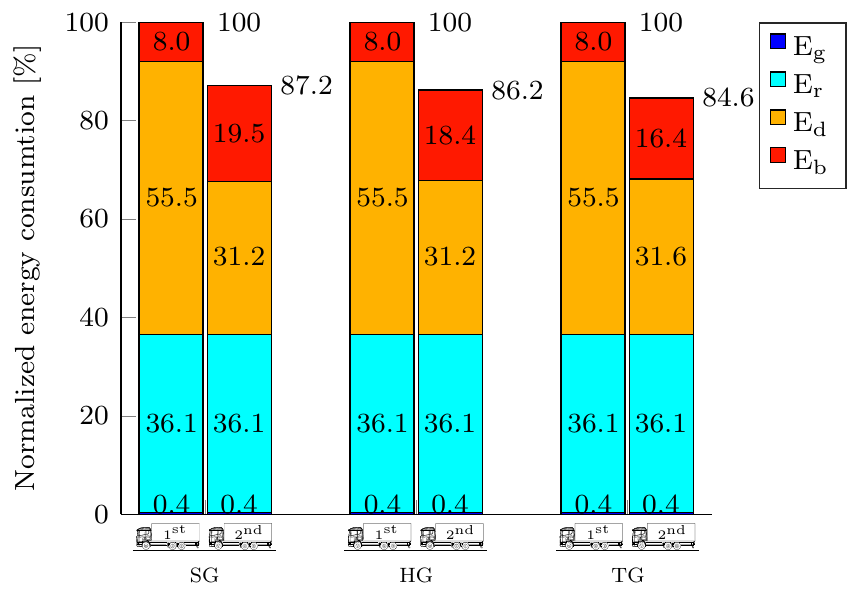}
\caption{Comparison of the energy consumed by each vehicle of an homogeneous platoon (i.e., $m_1=m_2=40$ t) for three different gap policies, namely space (SG), headway (HG) and time (TG) gap policies, while using CC as control strategy, driving along the $45$ km road displayed in Figure~\ref{fig:matlab_road_topology45}. For the plots explanation refer to the caption of Figure~\ref{fig:matlab_comparison_control_bars}.}
\label{fig:matlab_comparison_gaps_bars}
\end{center}
\end{figure}

\begin{figure}
\begin{center}
\includegraphics{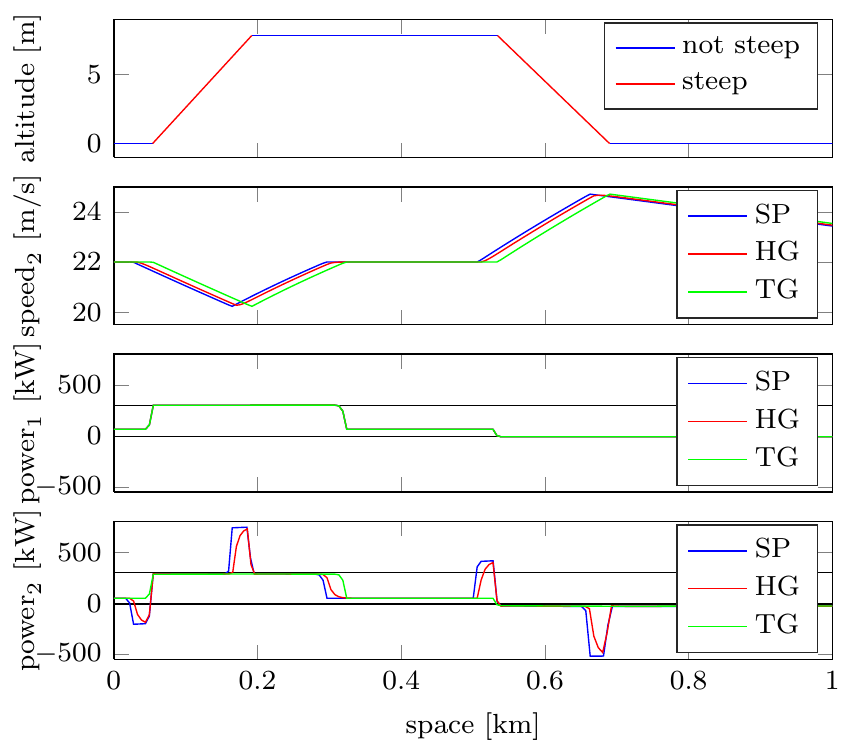}
\caption{Comparison of the behavior of a homogeneous platoon (i.e., \mbox{$m_1=m_2=40$ t)} for three different gap policies, namely space (SG), headway (HG) and time (TG) gap policies, while using CC as control strategy, driving over a synthetic hill. For the plots explanation refer to the caption of Figure~\ref{fig:matlab_comparison_control_plot}; note that the second plot shows only the speed trajectories of the second vehicle (the speed trajectory of the first vehicle coincides with that one of the second vehicle in the case of TG policy).}
\label{fig:matlab_comparison_gaps_plot}
\end{center}
\end{figure}

\begin{table}[h]
\caption{\rm Normalized fuel consumption of the vehicles in the platoon for different control strategies and gap policies. The fuel is normalized respect to the fuel consumed by the respective HDV driving alone using CC. For the acronyms explanation refers to Section~\ref{sec:fuel_efficiency_analysis} [\%].}
\label{tab:percentage_fuel_saved_2}
\begin{center}
\begin{tabular}{C|C C|C C| C C}
& \multicolumn{2}{c|}{SG} & \multicolumn{2}{c|}{HG} & \multicolumn{2}{c}{TG}  \\
& \multicolumn{2}{c|}{ \includegraphics[scale=1]{platoon_two_vehicle_table_num}} & 
\multicolumn{2}{c|}{ \includegraphics[scale=1]{platoon_two_vehicle_table_num}} & 
\multicolumn{2}{c}{\includegraphics[scale=1]{platoon_two_vehicle_table_num}}  \Bstrut\\
   \hline 
CC & $100.0$ & $88.6$ & $100.0$ & $87.7$ &$100.0$ & $86.3$ \Tstrut\\
LAC & $96.9$ & $82.7$ & $96.9$ & $81.9$ & $96.9$ & $80.6$ \\
CLAC & $97.0$ & $80.4$ & $97.0$ & $79.3$ & $97.0$ & $77.4$ \\
\end{tabular}
\end{center}
\end{table}

\section{Performance analysis of the vehicle controller} \label{sec:analysis_vehicle_controller}
In this section we analyze the performance of the vehicle controller layer (as presented in Section \ref{sec:trajectory_tracking} and shown in Figures \ref{fig_architecture} and \ref{fig_formulation_architecture}) by focusing on the safety aspect. The analysis is based on the simulation result displayed in Figure~\ref{fig:matlab_performance_braking} and Figure~\ref{fig:matlab_performance_braking_2}, where the leading HDV of a three vehicles platoon driving on a flat road brakes repeatedly with different braking profiles. Here we assume that the leading vehicle in the braking phases is manually driven and, therefore, the control system does not know a priori the braking profile. The considered vehicles are identical with the parameters as defined in Table~\ref{tab:vehicle_parameters}.

\subsection{Safety analysis} 
Here we focus on the safety analysis of the distributed vehicle controller layer and, in particular, we analyze the role of the safety constraint in various situations. In Figure~\ref{fig:matlab_performance_braking}, the leading vehicle is braking with deceleration of $1$, $2$ and $3$ m/s$^2$ for $0.9$ s at respectively $5$, $25$ and $55$ s. In the second plot of this figure, the effective distances and that ones that would activate the safety constraint (we will refer to it as the safety distance) are shown. First we can notice how, in line with our expectation, the second and third vehicles are braking (see the third plot) only when the effective distance touches the safety distance. In fact here we recall that, according to how the vehicle controller is designed (see Section~\ref{sec:model_predictive_control_forumation}), only the activation of the safety constraint or a braking request from the platoon coordinator can lead to a significant braking action. Consequently, during the first braking of $1$ m/s$^2$, both follower vehicles do not brake, despite the deviation of their states from the reference trajectories. During the second braking of $2$ m/s$^2$, instead, the safety constraint of the second vehicle are activated and therefore it requires a braking action. Finally, during the third braking of $3$ m/s$^2$, the safety constraints of both follower vehicles activate and therefore they both brake. Note that the safety constraint is designed such that fuel-efficiency has priority on driver comfort. In fact, in this case, in order to be fuel-efficient, the braking action is required only when the platoon is in a safety critical situation. However, a priori knowledge of the braking profile of the first vehicle (e.g. having a model of the driver or handling the braking action autonomously) would have allowed to have a smoother  and less intense braking action.

\begin{figure}
\begin{center}
\includegraphics{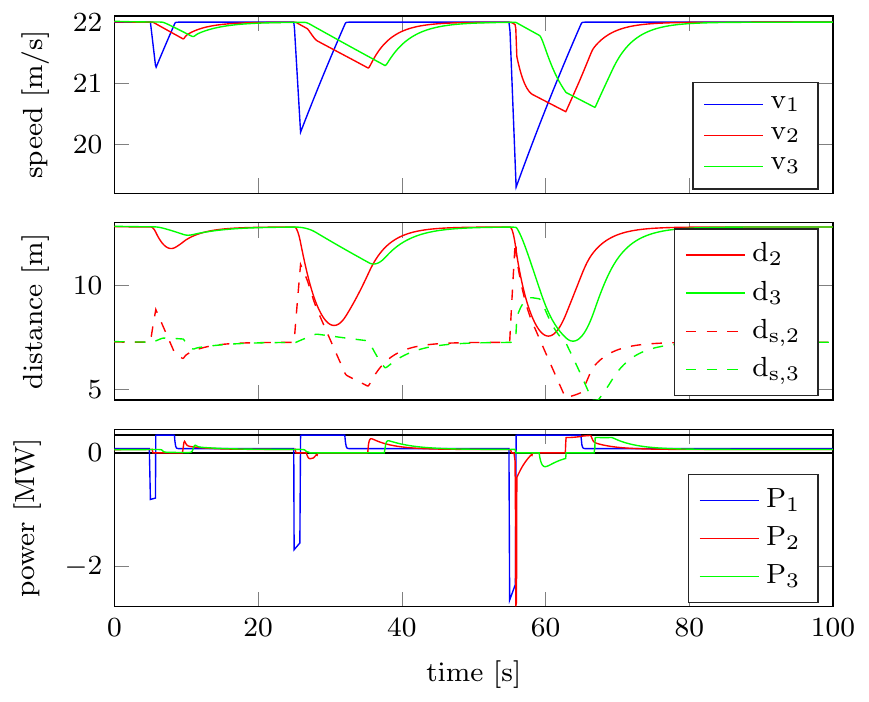}
\caption{Behavior of a three identical vehicles platoon driving on a flat road. The leading HDV  brakes three times at $5$, $25$ and $55$ s, with a braking deceleration of respectively $1$, $2$ and $3$ m/s$^2$ for $0.9$ s. The first plot shows the speed of the three vehicles: the second plot shows the distance between the vehicles and the respective safety distance computed using an adaptation of inequality \eqref{eq:safety_constraint}; the third plot shows the summation between the generated power by the engine and the braking systems of the vehicles.}
\label{fig:matlab_performance_braking}
\end{center}
\end{figure}

\begin{figure}
\begin{center}
\includegraphics{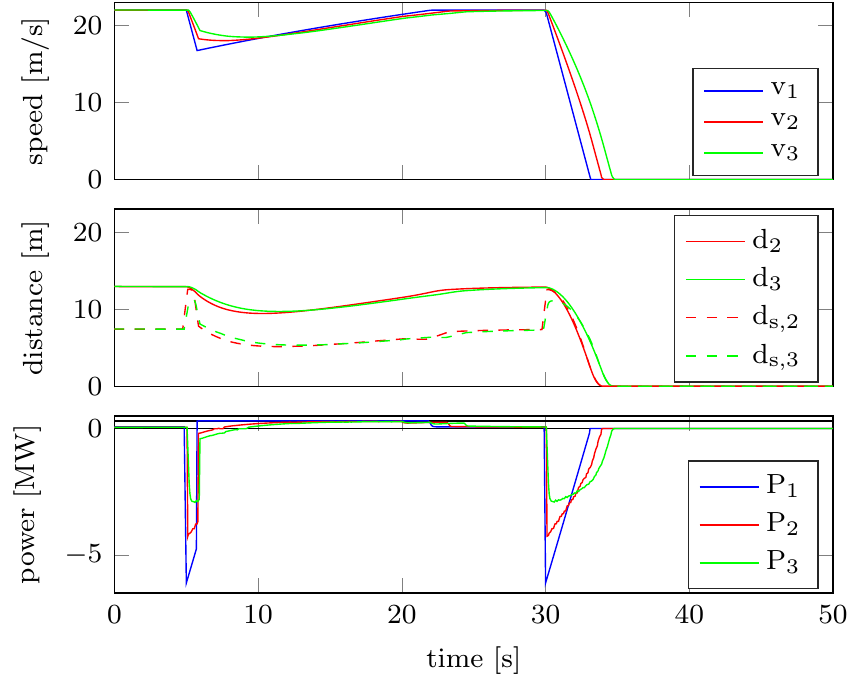}
\caption{Behavior of a three identical vehicles platoon driving on a flat road. The leading HDV  brakes a first time at $5$ s for $1$ s with a deceleration of $7$ m/s$^2$ and a second time at $25$ s with a deceleration of $7$  m/s$^2$ until arriving to full-stop. For the plot explanation refer to the caption of Figure~\ref{fig:matlab_performance_braking}.}
\label{fig:matlab_performance_braking_2}
\end{center}
\end{figure}

In Figure~\ref{fig:matlab_performance_braking_2}, we consider a more challenging scenario in which the first vehicle brakes with higher intensity, simulating an emergency situation. More precisely it brakes at $5$ s with a deceleration of $7$ m/s$^2$ for $1$ s and at $30$ s with the same deceleration until it arrives to full-stop. We can notice how, also in this scenario the safety constraint in each vehicle controller layer activates the braking action and guarantees no collision between the vehicles. 


\section{Performance analysis of the integrated system} \label{sec:analysis_integrated_system}
In this section we analyze the simulation results displayed in Figure~\ref{fig:matlab_integrated_system_simulation} of the platoon under the control of the integrated control architecture (i.e., platoon coordinator and vehicle controller). More precisely in this analysis we consider a platoon of three identical vehicles (whose parameters are defined in Table~\ref{tab:vehicle_parameters}) driving along the Sector A highlighted in Figure~\ref{fig:matlab_road_topology45}. This is the same sector for which the experimental results in \cite{phd_alam_2014} are displayed in Figure~\ref{fig:matlab__experiment} and analyzed in Section~\ref{sec:motivational_experiment}. 

At first glance, as expected from the platoon coordinator formulation, we can notice how all the vehicles follow the same speed and distance profiles in the space domain. Additionally, in order to follow such profiles, we can observe in the last plot how the second and third vehicle, thanks to the air drag reduction, need to generate less power than the leading vehicle. We now continue the analysis by focusing on the three segments highlighted in Figure~\ref{fig:matlab_integrated_system_simulation}:
\begin{itemize}
\item \textit{Segment 1:} due to the steep downhill, all vehicles are not able to maintain the constant speed without braking and, therefore, accelerate. However the platoon coordinator requires the leading vehicle to fuel slightly such that the follower vehicles can coast. In this case, the coordination role of the platoon coordinator allows to avoid braking action to all vehicles. 
\item \textit{Segment 2:} Since no gear shift is simulated the vehicles are able to maintain the time gap requirement during the uphill. 
\item \textit{Segment 3:} due to the longer downhill compare to the first one, the platoon exhibits a different behavior. First, the platoon coordinator requires all vehicles to decrease the speed to the minimum allowed (in this simulation it is set to $19$ m/s ) in order to hit the speed limit as late as possible. Second, since the speed limit is reached despite the decrease of speed at the beginning of the downhill, the platoon coordinator requires the first vehicle to coast and the follower vehicles to brake slightly to maximize the efficiency. In fact, in this case, to require the first vehicle to  fuel slightly and brake at the end of the downhill would be contradictory. 
\end{itemize}	
In conclusion the platoon, under the control of the integrated control architecture, shows a fuel-efficient and smooth behavior.

\begin{figure}
\begin{center}
\includegraphics{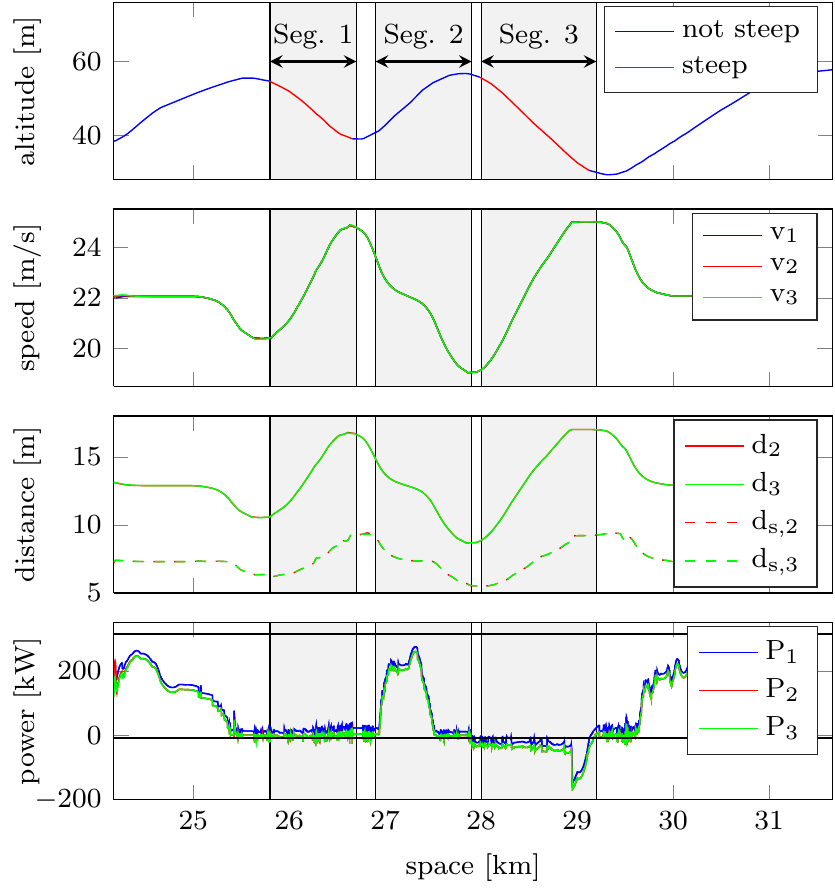}
\caption{Simulation results obtained using the proposed controller for a three-vehicle platoon while driving along the Sector A highlighted in Figure~\ref{fig:matlab_road_topology45}. The three vehicles are identical with parameters shown in Table~\ref{tab:vehicle_parameters}. The first plot shows the road topography. For the explanation of the other plots refer to the caption of Figure~\ref{fig:matlab_performance_braking}.}
\label{fig:matlab_integrated_system_simulation}
\end{center}
\end{figure}

\section{Conclusions and future works} \label{sec:conclusion}
\subsection{Conclusions}
In this paper we have presented a novel control architecture based on look-ahead control for fuel-efficient and safe HDV platooning. 

The use of a look-ahead control framework for HDV platooning has been first motivated by the analysis of real experiments. In particular in this analysis we concluded that the use topography information in order to predict the behavior of the vehicles and coordination between the vehicles can be beneficial for both fuel-efficiency and safety reasons. 

This led to the design of a novel control architecture for platooning. Such architecture is divided into two layers. A centralized higher layer, denoted as platoon coordinator, is responsible for the coordination of the platoon by defining a speed profile that is feasible and fuel-efficient for the entire platoon by exploiting preview topography information. Such speed profile is communicated to each block of the decentralized lower layer, denoted as vehicle controller layer. Within each vehicle controller a model predictive control routine tracks the reference speed profile and generates the real-time desired acceleration for the low-level vehicle controller. 

The performance of such control architecture has been evaluated through the analysis of numerical experiments. In details, the performance of the two layers has been studied both separately and in conjunction.

\subsection{Future works}
In the modeling of the vehicle powertrain we have assumed that the gear ratio can be chosen on a continuous interval on a unlimited span. However, this is not typically the case in commercial HDVs, where usually the transmission is handled by a gearbox that introduces fixed gear ratios and power losses during the gear shifts. Therefore in some future works we want to investigate how the presence of a gearbox should be managed in a optimal way. The optimal engine speed as function of the generated power shown in Figure~\ref{fig:fuel_model} and the knowledge of the current speed can be used to compute the instantaneous optimal gear ratio. However the power loss and the delay during the gear shift make the problem of when the HDVs should change gear (e.g., independently or simultaneously) and which gear they should engage not trivial. 

Secondly, we would like to investigate how external disturbances, as traffic ahead or a vehicle cutting in the platoon, can be handled in an autonomous way within the platoon controller framework. So far, in fact, such disturbances have been assumed to be handled manually by the drivers. However the prediction of local traffic would allow the platoon to move fuel-efficiently and safely in it.

\section*{Acknowledgment}
The authors gratefully acknowledge the European Union's Seventh Framework Programme within the project \mbox{COMPANION}, the Swedish Research Council and the Knut and Alice Wallenberg Foundation for their financial support. The authors also thank Dr.\ Assad Alam and Dr.\ Henrik Pettersson from Scania for the fruitful discussions.

\bibliographystyle{plain}        
	 
\end{document}